\documentclass[10pt,journal]{IEEEtran}

\usepackage{amsmath,amssymb,amsfonts}
\usepackage{amsthm}
\usepackage{bm}

\usepackage{graphicx}
\usepackage[caption=false,font=footnotesize]{subfig}

\usepackage{booktabs}
\usepackage{multirow}
\usepackage{array}

\usepackage{algorithm}
\usepackage{algorithmic}

\usepackage{enumerate}

\usepackage{cite}

\usepackage{url}

\usepackage{balance}

\usepackage{hyperref}
\hypersetup{
  colorlinks=true,
  linkcolor=blue,
  citecolor=blue,
  urlcolor=black,
}

\newtheorem{assumption}{\textbf{Assumption}}
\newtheorem{proposition}{\textbf{Proposition}}
\newtheorem{theorem}{\textbf{Theorem}}

\theoremstyle{remark}
\newtheorem{remark}{\textbf{Remark}}

\expandafter\def\expandafter\normalsize\expandafter{%
  \normalsize%
  \setlength\abovedisplayskip{4pt}%
  \setlength\belowdisplayskip{4pt}%
  \setlength\abovedisplayshortskip{2pt}%
  \setlength\belowdisplayshortskip{2pt}%
}

\begin{document}

\title{Site-Specific Learning for Low-Overhead Multi-User MIMO Beamforming}

\author{Cheng-Jie Zhao, Zhaolin Wang,~\IEEEmembership{Member,~IEEE},
  Zongyao Zhao,~\IEEEmembership{Member,~IEEE},
  and Yuanwei Liu,~\IEEEmembership{Fellow,~IEEE}%
  \vspace{-0.6cm}
  \thanks{The authors are with the Department of Electrical and Computer
Engineering, The University of Hong Kong, Hong Kong
(e-mail: chengjie\_zhao@connect.hku.hk, zhaolin.wang@hku.hk,
zongyao@hku.hk, yuanwei@hku.hk).}}
\maketitle

\begin{abstract}
A low-overhead site-specific multi-user multiple-input multiple-output (MU-MIMO) beamforming framework is proposed. Conventional limited-feedback MU-MIMO relies on channel state information reference signal (CSI-RS) transmission and user feedback before grouping and beamforming, which requires substantial online overhead when the antenna dimension and candidate-user pool are large. To reduce this burden, the proposed framework exploits site-specific information (SSI), which captures local radio propagation features. By learning the mapping from low-overhead beam-domain observations to effective transmit spatial subspaces of users, the BS can infer inter-user separability before high-resolution CSI acquisition and construct a compact group-level CSI acquisition subspace for the selected users. This site-specific design can be implemented within the standard limited-feedback procedure using synchronization signal block (SSB)-based reference signal received power (RSRP) fingerprints for subspace inference and CSI-RS feedback for low-dimensional CSI refinement. Extensive numerical results demonstrate that the proposed framework can identify compatible user groups before CSI-RS acquisition, preserve most scheduled-user channel energy in a compact group subspace, and achieve higher effective rates than conventional systems with significantly lower overhead and user-side processing burden.
\end{abstract}

\vspace{-0.3cm}
\section{Introduction}
\IEEEPARstart{M}{ULTI-USER} multiple-input multiple-output (MU-MIMO) is a fundamental physical-layer technique for improving the spectral efficiency of modern cellular networks~\cite{Larsson2014MassiveMIMO,Heath2016mmWave}. By exploiting the spatial degrees of freedom provided by large antenna arrays, a base station (BS) can transmit multiple data streams over the same time-frequency resource and support dense downlink traffic without requiring proportional growth in bandwidth or time resources. The achievable multiplexing gain depends not only on the number of transmitted streams, but also on the spatial separability of the users served on the same resource block. When their channels are sufficiently distinguishable at the BS array, spatial beamforming can suppress inter-user interference (IUI) and convert the available spatial degrees of freedom into rate gain. When the scheduled users are strongly coupled, the effective channel matrix becomes poorly conditioned and additional streams may provide limited spectral-efficiency improvement. This observation leads to two coupled tasks in MU-MIMO systems: user scheduling and user multiplexing. User scheduling determines which users should share the same time-frequency resource, while user multiplexing designs the beamforming vectors that separate the scheduled users and suppress IUI.
\\
\indent The user-multiplexing task has been extensively studied through MU-MIMO transmission optimization. Depending on the system model, the BS jointly designs beamforming matrices, transmit powers, antenna positions, or other controllable variables to improve sum-rate, energy efficiency, or other quality-of-service satisfaction~\cite{Christensen2008WMMSE,Sun2018NOMAMIMO,Zhao2026PRAFD}. These studies show how additional transmit-side degrees of freedom can be exploited in a configured MU link and have led to widely used tools such as weighted minimum mean square error (MMSE) optimization and successive convex approximation. Their starting point is a specified transmission configuration, so they focus on extracting performance after the users and resources have already entered the design problem. They do not resolve the complementary scheduling task of identifying which users have sufficiently separable spatial channels for joint MU transmission.
\\
\indent User selection and grouping methods address this scheduling task more directly. Semi-orthogonal user selection (SUS), for example, greedily selects users from instantaneous channel state information (CSI) and is asymptotically optimal for zero-forcing (ZF) beamforming when the candidate pool is large~\cite{Yoo2006SUS}. Beyond SUS, CSI-domain MU scheduling and linear beamforming studies commonly use channel correlation, projected channel gain, or effective SINR to avoid pairing users that would create strong IUI after beamforming~\cite{Spencer2004ZF,Wagner2012RZF}. This line of work makes clear that spatial separability is a key part of MU-MIMO operation. The limitation is that the compatibility metric itself is computed from instantaneous CSI of the candidate users. Since such CSI must first be obtained through reference signal (RS) transmission and subsequent estimation/feedback, CSI-domain grouping requires acquisition resources to be spent before the BS can decide which users are suitable for MU transmission. This becomes costly when the candidate-user pool is large and time-frequency and hardware resources are limited.
\\
\indent Existing work has therefore developed several ways to reduce or structure the CSI acquisition burden. One route is statistical or covariance-based MU-MIMO. Joint spatial division and multiplexing (JSDM) groups users according to long-term covariance eigenspaces and reduces the dimension of subsequent MU processing~\cite{Adhikary2013JSDM}. Related subspace analysis also shows that beamforming designs often admit compact low-dimensional representations~\cite{Zhao2025LDS}. These approaches demonstrate that long-term spatial subspace information can reduce the dimension of MU-MIMO acquisition and beamforming. However, they usually assume that the relevant covariance matrices or group eigenspaces are already known from long-term estimation. This leaves open how a BS should obtain useful group-level subspace information from the low-overhead measurements available before dedicated CSI acquisition.
\\
\indent Another route is limited-feedback CSI acquisition. Limited feedback requires users measure downlink RSs, compress the observed CSI, and report it to the BS~\cite{Love2008LimitedFeedback,Jindal2006Feedback}. In fifth-generation (5G) new radio (NR) systems, CSI reporting is supported by several codebook families, such as Type-I, Type-II, and port-selection codebooks, which provide different tradeoffs between feedback overhead and beamforming accuracy~\cite{3GPP38211,3GPP38214,3GPP38215}. Among them, the Type-II codebook is one of the most expressive options, representing the channel with selected beam indices and complex coefficients and thereby improving beamforming accuracy in rich multipath channels~\cite{Giordani2019BeamManagementNR}. Nevertheless, the gain in feedback resolution comes with additional CSI-RS measurement, codebook search, and reporting burden. As antenna arrays and candidate-user pools grow, this burden scales with both the CSI-RS dimension and the number of users to be screened, making CSI acquisition and feedback a major bottleneck in limited-feedback MU-MIMO.
\\
\indent Site-specific learning provides a solution to address the overhead bottleneck in limited-feedback systems by exploiting \emph{site-specific information (SSI)}, i.e., the propagation regularities induced by a fixed deployment \cite{Zeng2020CKM}. In a fixed site, the surrounding geometry, dominant scatterers, blockage patterns, and BS deployment vary slowly compared with instantaneous fading, making the propagation distribution learnable from site-specific data. Along this direction, generative site-specific beamforming has been developed to synthesize environment-conditioned beams from learned deployment-specific channel distributions~\cite{sim}. A site-specific Type-II feedback framework was proposed in~\cite{prior_singleuser}, where synchronization signal block (SSB)-based reference signal received power (RSRP) fingerprints are used to infer a user's dominant transmit subspace before dedicated CSI-RS transmission. Cross-site pretraining and target-site calibration were further studied in SiFo~\cite{SiFo} to improve the scalability of site-specific CSI feedback across different deployments. Related studies also support this principle through site-specific probing for fast beam alignment~\cite{Heng2022SiteSpecificProbing}, RSRP-adaptive codebook design~\cite{Ning2023RSRPCodebook}, and deep-learning-assisted beam selection~\cite{Alkhateeb2018DLBeam}. Taken together, these studies show that RSRP fingerprints and other beam-domain measurements can serve as low-overhead observations of SSI before high-resolution CSI feedback, thereby reducing online sensing, search, and feedback overhead. However, they mainly target beam alignment, codebook selection, or single-link CSI acquisition rather than MU-level site-specific beamforming.
\\
\indent Moving from single-link CSI acquisition to MU-MIMO beamforming changes the role of SSI. In single-user feedback, SSI mainly helps determine how to acquire one user's CSI with fewer online resources. In MU transmission, SSI must provide additional guidance before detailed CSI is available, including which users are suitable for spatial multiplexing and how the limited CSI-RS resources should be shared by the scheduled group. This is more challenging because users with individually strong channels may still suffer severe IUI when their spatial signatures overlap, while independently acquiring each scheduled user's CSI would largely offset the overhead reduction enabled by site-specific learning. Therefore, MU-oriented SSI must be converted into group-aware information that supports both compatibility inference and shared CSI acquisition. In this paper, we implement this idea by extending the RSRP-based predictor developed in the single-user site-specific limited feedback framework \cite{prior_singleuser} to infer per-user dominant spatial support from low-overhead SSB measurements, and then use the inferred support to select compatible users, construct a compact group CSI-RS acquisition space, and then perform corresponding MU beamforming. The main contributions are summarized as follows.
\begin{itemize}
	\item We develop a site-specific MU-MIMO feedback framework with three stages: 1) SSB-based initial access, 2) CSI-RS-based CSI refinement, and 3) low-dimensional regularized ZF (RZF) MU beamforming. By exploiting the SSB-stage RSRP fingerprints for MU grouping and group-specific CSI-RS design, the framework reduces CSI acquisition and feedback overhead while preserving MU transmission performance, thereby improving the final effective rate.

  	\item We propose pre-CSI MU compatibility inference and group selection from RSRP-predicted dominant subspaces. The method evaluates user compatibility before CSI-RS acquisition and forms a scheduled group through a low-complexity greedy rule. We further relate the resulting compatibility metric to CSI-domain channel coupling, providing a theoretical basis for the proposed grouping method.

	\item We design a low-dimensional group CSI acquisition approach for the selected users. The proposed construction compresses the scheduled users' predicted subspaces into a compact group-shared CSI-RS subspace, preserving user channel energy, improving effective-channel conditioning for RZF beamforming, and reducing the complexity of MU beamforming matrix computation.

  	\item Extensive experiments on DeepMIMO scenarios validate the proposed framework from compatibility inference to MU transmission. The results show that the proposed method can identify compatible user groups before CSI-RS acquisition, preserve most scheduled-user channel energy, and achieve higher effective rates than conventional limited-feedback baselines.
\end{itemize}

\indent The remainder of this paper is organized as follows. Section~\ref{sec:system} introduces the system model and formulates the optimization problem. Section~\ref{sec:pipeline} reviews the conventional limited-feedback pipeline and introduces the proposed site-specific MU-MIMO framework. Section~\ref{sec:compatibility} develops pre-CSI MU compatibility inference and greedy group selection. Section~\ref{sec:group_csi} presents the group CSI-RS subspace construction, effective-channel feedback, and low-dimensional RZF beamforming. Section~\ref{sec:results} reports numerical results on DeepMIMO scenarios. Section~\ref{sec:conclusion} concludes the paper.
\\ \indent \emph{Notation:} Scalars, vectors, matrices, and sets are denoted by italic letters, boldface lowercase letters, boldface uppercase letters, and calligraphic letters, respectively; $|\mathcal{A}|$ denotes the cardinality of $\mathcal{A}$. $(\cdot)^T$ and $(\cdot)^H$ denote the transpose and Hermitian transpose. $\|\mathbf{x}\|$, $\|\mathbf{X}\|_F$, and $\|\mathbf{X}\|_2$ denote the Euclidean norm, Frobenius norm, and spectral norm, respectively. $\mathrm{tr}(\cdot)$ is the trace operator, $\mathbf{I}_n$ the $n\times n$ identity matrix, and $\mathbb{E}[\cdot]$ the expectation operator. $[N]\triangleq\{1,\ldots,N\}$ denotes the index set of the first $N$ positive integers. $\mathcal{CN}(\bm{\mu}, \bm{\Sigma})$ denotes the circularly symmetric complex Gaussian distribution with mean $\bm{\mu}$ and covariance $\bm{\Sigma}$.
\vspace{-0.3cm}
\section{System Model and Problem Formulation}
\label{sec:system}
We consider a single-cell downlink MU-MIMO system, where a BS equipped with $N_t$ fully digital transmit antennas serves a candidate pool $\mathcal{K}_{\rm{full}}$ of single-antenna users. The downlink channel is assumed to be block fading, i.e., approximately constant within each coherence interval. Hence, one CSI acquisition period is carried out in each coherence interval. In this paper, $\mathcal{K}_{\rm full}$ is assumed to be determined by higher-layer mechanisms such as connection management and long-term scheduling, which is outside the scope of this work. During each coherence interval, the BS selects a group of $K$ users, which are in subset $\mathcal{K} \subseteq \mathcal{K}_{\rm{full}}$, for simultaneous spatial multiplexing, and assigns each scheduled user $k \in \mathcal{K}$ a dedicated beamforming vector $\mathbf{w}_k \in \mathbb{C}^{N_t\times 1}$.
\vspace{-0.5cm}
\subsection{Channel and Signal Model}
Following the geometric scattering representation~\cite{Ayach2014SpatiallySparse}, the downlink channel vector of user $k$, denoted by $\mathbf{h}_k \in \mathbb{C}^{N_t}$, can be modeled as a superposition of $L$ dominant paths:
\begin{equation}
  \mathbf{h}_k = \sum_{l=1}^{L} \alpha_{k,l}\, \mathbf{a}(\phi_{k,l}),
  \label{eq:channel}
\end{equation}
where $\alpha_{k,l} \sim \mathcal{CN}(0,\sigma_{k,l}^2)$ is the complex gain of the $l$-th path and $\mathbf{a}(\phi_{k,l}) \in \mathbb{C}^{N_t}$ is the transmit steering vector at the angle of departure $\phi_{k,l}$. A uniform linear array (ULA) is assumed in this paper with inter-element spacing $d$ and wavelength $\lambda$. Hence, the steering vector takes the form
\begin{equation}
	\mathbf{a}(\phi_{k,l}) = \mathbf{a}(u_{k,l}) = \frac{1}{\sqrt{N_t}} \bigl[1,\, e^{j2\pi u_{k,l}},\, \ldots,\, e^{j2\pi(N_t-1)u_{k,l}}\bigr]^T,
	\label{eq:steering}
\end{equation}
where $u_{k,l} \triangleq (d/\lambda)\sin(\phi_{k,l})$ is the normalized spatial frequency. By assuming that $L$ path gains $\left\lbrace \alpha_{k,l}\right\rbrace_{l \in [L]}$ are independent zero-mean complex Gaussian, $\mathbf{h}_k$ is zero-mean complex Gaussian with channel covariance $\bm{\Sigma}_k \triangleq \mathbb{E}[\mathbf{h}_k\mathbf{h}_k^H]$ of rank $L \ll N_t$~\cite{Adhikary2013JSDM}. We capture this low-rank structure through the following subspace model.
\begin{assumption}[\emph{Low-Rank Channel Model}]
\label{asm:channel}\normalfont
The channel of user $k$ lies in a $Q$-dimensional subspace with $Q \ll N_t$, i.e., $\mathbf{h}_k = \mathbf{U}_k\mathbf{g}_k$, where $\mathbf{U}_k \in \mathbb{C}^{N_t \times Q}$ has orthonormal columns and $\mathbf{g}_k \sim \mathcal{CN}(\mathbf{0}, \mathbf{R}_k)$ with $\mathrm{tr}(\mathbf{R}_k) = \beta_k$ being the power of channel $\mathbf{h}_k$.
\end{assumption}
\indent Under \textbf{Assumption~\ref{asm:channel}}, the channel covariance factorizes as
\begin{equation}
	\bm{\Sigma}_k = \mathbf{U}_k\mathbf{R}_k\mathbf{U}_k^H,
	\label{eq:covariance}
\end{equation}
which is of rank $Q$ with all energy confined to the subspace $\mathcal{U}_k=\mathrm{span}(\mathbf{U}_k)$.
\begin{remark}[\emph{Geometric instantiation}] \normalfont
The geometric scattering model~\eqref{eq:channel} instantiates \textbf{Assumption~\ref{asm:channel}}. Compactly, collecting $\mathbf{A}_k = [\mathbf{a}(u_{k,1}),\ldots,\mathbf{a}(u_{k,L})] \in \mathbb{C}^{N_t \times L}$ and $\boldsymbol{\alpha} = [\alpha_{k,1},\ldots,\alpha_{k,L}] \in \mathbb{C}^{L}$ yields $\mathbf{h}_k = \mathbf{A}_k\bm{\alpha}_k$. Therefore, the covariance matrix is given by $\bm{\Sigma}_k = \mathbf{A}_k\mathbf{D}_k\mathbf{A}_k^H$, where $\mathbf{D}_k = \mathrm{diag}(\sigma_{k,1}^2,\ldots,\sigma_{k,L}^2)$. Since $\bm{\Sigma}_k$ has rank at most $L$, the subspace dimension satisfies $Q \leq L$. In particular, when the $L$ normalized spatial frequencies $\{u_{k,l}\}$ are well-separated~\cite{Adhikary2013JSDM,3GPPTR38901}, we have $\mathbf{A}_k^H\mathbf{A}_k \approx \mathbf{I}_L$~\cite{Ayach2014SpatiallySparse} and $\bm{\Sigma}_k$ has full rank $L$, so the $Q = L$ instance of \textbf{Assumption~\ref{asm:channel}} applies with $ \mathrm{span}(\mathbf{A}_k)=\mathcal{U}_k$ and $\mathbf{R}_k \approx \mathbf{D}_k$.
\end{remark}

For any low-dimensional subspace $\mathcal{V}\subseteq \mathbb{C}^{N_t}$ with basis $\mathbf{V}\in\mathbb{C}^{N_t\times d_{\mathcal V}}$, representing $\mathbf{h}_k$ within $\mathcal{V}$ retains only its projected component. The retained channel energy is measured by the \emph{CSI-capture efficiency}~\cite{SiFo,prior_singleuser}:
\begin{equation}
	\eta_k(\mathcal{V}) = \frac{\|\mathbf{P}_{\mathcal{V}}\mathbf{h}_k\|^2}{\|\mathbf{h}_k\|^2} \in [0,1],
	\label{eq:capture_eff}
\end{equation}
where $\mathbf{P}_{\mathcal{V}}\triangleq \mathbf{V}(\mathbf{V}^H\mathbf{V})^{-1}\mathbf{V}^H$ denotes the orthogonal projector onto $\mathcal{V}$. Under \textbf{Assumption~\ref{asm:channel}}, the channel $\mathbf{h}_k$ lies in the $Q$-dimensional dominant subspace $\mathcal{U}_k$. Therefore, projection onto $\mathcal{U}_k$ is lossless, i.e., $\eta_k(\mathcal{U}_k)=1$ with $\mathbf{P}_{\mathcal{U}_k}=\mathbf{U}_k\mathbf{U}_k^H$, whereas any subspace that misses part of $\mathcal{U}_k$ discards the corresponding channel component.
\\ \indent The BS transmits a single data stream to each scheduled user. At each channel use, the data symbol for user $k$ is denoted by $s_k \sim \mathcal{CN}(0, P_t/K)$, with symbols independently generated across users and channel uses and with equal per-stream power $P_t/K$. Stacking the beamforming vectors into the beamforming matrix $\mathbf{W} = [\mathbf{w}_1,\dots,\mathbf{w}_K] \in \mathbb{C}^{N_t \times K}$ and the symbols into the symbol vector $\mathbf{s} = [s_1,\dots,s_K]^T \in \mathbb{C}^{K\times 1}$, the transmitted signal is given by $\mathbf{x} = \mathbf{W}\mathbf{s}$. The received signal at user $k$ is then represented as
\begin{equation}
	y_k = \mathbf{h}_k^H \mathbf{w}_k s_k + \sum_{j \in \mathcal{K},\, j \neq k} \mathbf{h}_k^H \mathbf{w}_j s_j + n_k,
	\label{eq:received}
\end{equation}
where $n_k \sim \mathcal{CN}(0, \sigma_n^2)$ is additive white Gaussian noise independent of the data symbols. The signal-to-interference-plus-noise ratio (SINR) at user $k$ is given by
\begin{equation}
  {\rm SINR}_k(\mathbf{W}) = \frac{(P_t/K)\,|\mathbf{h}_k^H \mathbf{w}_k|^2}{\sigma_n^2 + (P_t/K) \sum_{j \in \mathcal{K},\, j \neq k} |\mathbf{h}_k^H \mathbf{w}_j|^2}.
  \label{eq:sinr}
\end{equation}
The corresponding achievable rate is
\begin{equation}
	R_k(\mathbf{W}) = \log_2(1 + {\rm SINR}_k(\mathbf{W})).
\end{equation}

\subsection{Problem Formulation}
\label{sec:problem}
Within a single coherence interval, the BS aims to maximize the total achievable rate by jointly selecting a compatible user group $\mathcal{K}$ and designing the corresponding beamforming matrix $\mathbf{W}$. The sum-rate maximization problem can be formulated as follows:
\begin{equation}
	\max_{\mathcal{K},\;\mathbf{W}}\;  \sum_{k \in \mathcal{K}} R_k(\mathbf{W}), \; \mathrm{s.t.}\  \|\mathbf{w}_k\| = 1, \mathcal{K} \subseteq \mathcal{K}_{\rm full}, |\mathcal{K}|=K.
  \label{eq:opt_prob}
\end{equation}
Solving~\eqref{eq:opt_prob} directly is intrinsically intractable for two main reasons. First, the objective depends on the instantaneous downlink CSI of all users in the candidate pool, i.e., $\{\mathbf{h}_k\}_{k \in \mathcal{K}_{\rm full}}$. Without the CSI, the BS cannot evaluate the user rates, compare the spatial separability of different user groups, and finally compute the corresponding beamforming matrix. Second, even if instantaneous CSI were available, the joint optimization over user grouping and beamforming remains computationally challenging. Concretely, the group-selection is combinatorial, since the BS must choose a subset from $\mathcal{K}_{\rm full}$, and the beamforming design for sum-rate maximization is generally non-convex due to the IUI.
\\ \indent In practical 5G NR systems, the difficulty of instantaneous CSI acquisition is addressed through limited-feedback CSI reporting. Then, the original problem in \eqref{eq:opt_prob} is turned into a sequential CSI acquisition, feedback, grouping, and beamforming joint design problem. The next section reviews this limited-feedback pipeline and introduces how the SSI can be incorporated to reduce the CSI acquisition burden.
\vspace{-0.3cm}
\section{Limited Feedback Pipeline and Proposed Site-Specific Framework}
\label{sec:pipeline}
This section reviews the conventional NR limited-feedback pipeline and introduces the proposed site-specific framework, contrasting how each addresses the CSI dependency identified in problem~\eqref{eq:opt_prob}.
\vspace{-0.3cm}
\subsection{Conventional Limited Feedback Pipeline}
In practical NR systems, MU-MIMO transmission is not realized by directly solving problem~\eqref{eq:opt_prob} over all possible user groups and beamforming matrices. Instead, it follows a limited-feedback pipeline that first obtains beam-level measurements for initial access, then refines the downlink CSI through CSI-RS transmission and user feedback, and finally selects compatible users for joint beamforming. Specifically, we abstract and separate the conventional pipeline into the following three stages.
\\ \indent
\textbf{Stage 1 (SSB-based Initial Access):} The procedure starts with SSB beam sweeping, where the BS transmits a predefined set of synchronization beams $\mathbf{B}=[\mathbf{b}_1,\ldots,\mathbf{b}_B]\in\mathbb{C}^{N_t\times B}$, which are typically discrete Fourier transform (DFT) beams, and each user measures the RSRP of these beams for initial access and coarse directional acquisition. For the $i$-th SSB beam, the received signal at the $k$-th user is
\begin{equation}
  \mathbf{y}_{k,i}^{\rm SSB} = \sqrt{P_{\rm SSB}}\,\mathbf{h}_k^H \mathbf{b}_i\, \mathbf{s}_{\rm SSB} + \mathbf{n}_{k,i}^{\rm SSB},
  \label{eq:ssb_received}
\end{equation}
where $P_{\rm SSB}$ is the SSB transmit power, $\mathbf{s}_{\rm SSB}\in\mathbb{C}^{L_s \times 1}$ collects the $L_s$ SSB symbols, and $\mathbf{n}_{k,i}^{\rm SSB}\sim\mathcal{CN}(\mathbf{0}, \sigma_n^2\mathbf{I}_{L_s})$\footnote{SSB symbols occupy specific time-frequency resource elements in 5G NR. For analytical compactness they are aggregated into a vector and their time-frequency indices are not separately tracked.}. Since SSB is a cell-wide broadcast, all users measure from the same downlink transmission and no IUI arises. For each SSB beam, the user averages the received signal power over the corresponding SSB symbols and obtains the RSRP as
\begin{equation}
	\mathrm{RSRP}_{k,i} \triangleq \frac{1}{L_s}\bigl\|\mathbf{y}_{k,i}^{\rm SSB}\bigr\|^2 .
	\label{eq:rsrp_linear}
\end{equation}
Following~\cite{sim}, the dB-domain RSRP can be modeled as
\begin{equation}
	r_{k,i} = 10\log_{10}(\mathrm{RSRP}_{k,i}) = r_{k,i}^0 + n_{k,i},
	\label{eq:rsrp}
\end{equation}
where $r_{k,i}^0$ is the noise-free RSRP and $n_{k,i}$ is the Gaussian perturbation capturing thermal noise and shadowing variability. Detailed expressions for $r_{k,i}^0$ and $n_{k,i}$ can refer to Eq. (7), Eq. (8a), and Eq. (8b) in \cite{sim}. Collecting the $B$ measurements over all SSB beams gives the beam-domain RSRP fingerprint $\mathbf{r}_k = [r_{k,1},\ldots,r_{k,B}]^T \in \mathbb{R}^{B \times 1}$ for user $k$. This fingerprint provides coarse beam-level information and mainly supports initial access and beam management. We assume only the index and RSRP of the strongest beam is reported in the conventional pipeline.
\\ \indent
\textbf{Stage 2 (CSI-RS-based CSI Refinement):} By definition, the RSRP fingerprint measured with SSB beams does not directly provide the phase-coherent channel direction required to construct high-resolution downlink beamforming vectors. To support data transmission, NR systems therefore proceed to a CSI refinement stage using CSI-RS~\cite{Giordani2019BeamManagementNR,3GPP38214,3GPP38215, prior_singleuser}. Type-II feedback is considered in this paper. Detailed analysis regarding Type-I, Type-II, and port selection feedback schemes is provided in~\cite{prior_singleuser}.
\\ \indent  
The CSI refinement stage with Type-II feedback consists of three steps. \emph{(i)} The BS transmits non-precoded CSI-RSs over $L_c \geq N_t$ pilot channel uses, with per-use transmit vector drawn from the $N_t$-port pilot matrix $\mathbf{S}_{\rm CSI}\in\mathbb{C}^{N_t\times L_c}$ satisfying $\mathbf{S}_{\rm CSI}\mathbf{S}_{\rm CSI}^H = \mathbf{I}_{N_t}$. The aggregate received signal at user $k$ is
\begin{equation}
	\mathbf{y}_k^{\rm CSI} = \sqrt{P_{\rm CSI}}\,\mathbf{S}_{\rm CSI}^H\mathbf{h}_k + \mathbf{n}_k^{\rm CSI} \in \mathbb{C}^{L_c \times 1},
	\label{eq:csi_received}
\end{equation}
where $P_{\rm CSI}$ is the CSI-RS transmit power and $\mathbf{n}_k^{\rm CSI}\sim\mathcal{CN}(\mathbf{0},\sigma_n^2\mathbf{I}_{L_c})$. Similar to SSB, the non-precoded CSI-RS is transmitted to all candidate users in $\mathcal{K}_{\rm full}$ simultaneously, all candidates measure from the same downlink transmission and no IUI arises. \emph{(ii)} Each candidate user estimates the downlink channel $\mathbf{h}_k$ from~\eqref{eq:csi_received} via standard channel estimators, e.g., least-squares (LS). Perfect channel estimation is assumed here to isolate the effect of codebook compression and feedback. \emph{(iii)} Instead of feeding back the full channel vector, the user selects $Q$ dominant columns from the oversampled DFT codebook $\mathbf{D}\in\mathbb{C}^{N_t\times O_D N_t}$ with oversampling rate $O_D$~\cite{3GPP38214,3GPP38215} to approximately span the dominant $Q$-dimensional subspace $\mathcal{U}_k$ in \textbf{Assumption~\ref{asm:channel}}. The beam index set selected by user $k$ is denoted by
\begin{equation}
	\mathcal{S}_k = \underset{\mathcal{S}\subseteq[O_D N_t],\;|\mathcal{S}|=Q}{\arg\max}\;\bigl\|\mathbf{D}_{\mathcal{S}}^H\mathbf{h}_k\bigr\|^2,
	\label{eq:beam_select}
\end{equation}
where $\mathbf{D}_{\mathcal{S}}\in\mathbb{C}^{N_t\times Q}$ collects the columns of $\mathbf{D}$ indexed by $\mathcal{S}\subseteq[O_D N_t]$. The user then computes the combining coefficients via LS projection onto the selected subspace:
\begin{equation}
	\boldsymbol{\alpha}_k = \mathbf{D}_{\mathcal{S}_k}^\dagger\,\mathbf{h}_k \in \mathbb{C}^Q.
	\label{eq:typeii_coeff}
\end{equation}
Subsequently, $(\mathcal{S}_k,\hat{\boldsymbol{\alpha}}_k)$ are reported back to the BS~\cite{3GPP38214,3GPP38215}, with which the BS reconstructs $\hat{\mathbf{h}}_k = \mathbf{D}_{\mathcal{S}_k}\hat{\boldsymbol{\alpha}}_k$.
\\ \indent
\textbf{Stage 3 (Full-Dimensional MU Beamforming):} After receiving the CSI reports, in the single-user case, the BS directly applies maximum ratio transmission (MRT), setting $\mathbf{w}_k^{\rm single} = \hat{\mathbf{h}}_k / \|\hat{\mathbf{h}}_k\|$ to maximize the received signal-to-noise ratio (SNR). In the multi-user case, however, simultaneous transmission to multiple users introduces IUI, whose severity depends on how well the channels of the co-scheduled users are separated. The BS must therefore first assess inter-user separability among candidates before committing to a group and a beamforming matrix. Concretely, it compares the reconstructed channel directions and selects a compatible group, for example using SUS~\cite{Yoo2006SUS}. Given the selected group $\mathcal{K}$ and its reconstructed CSI, the BS computes a full-dimensional beamforming matrix $\mathbf{W}$ for joint downlink transmission. In this paper, RZF~\cite{Peel2005RZF} is used as the reference full-dimensional beamforming method, so the user-specific beamforming vectors $\{\mathbf{w}_k\}$ are obtained from the RZF solution.
\vspace{-0.5cm}
\subsection{Problem Reformulation for Limited Feedback}
\label{sec:overhead}
The limited-feedback pipeline adopted in NR systems resolves the instantaneous CSI dependency in problem~\eqref{eq:opt_prob} through sequential CSI-RS transmission and corresponding feedback, while practical user grouping and linear beamforming replace the joint non-convex optimization. Under this pipeline, the overhead incurred by CSI acquisition and feedback critically affects the eventually achieved effective rate, since the time-frequency resources occupied by CSI-RS transmission and feedback cannot be used for data delivery. Concretely, the online overhead per coherence interval, measured in channel uses, consists of the SSB-based beam-management overhead $T_{\rm SSB}$ averaged over its update period and the CSI-RS transmission and feedback overhead $T_{\rm CSI\text{-}RS}$. It can be represented as
\begin{equation}
	T_o = \zeta_{\rm BM}T_{\rm SSB} + T_{\rm CSI\text{-}RS},
	\label{eq:overhead}
\end{equation}
where $\zeta_{\rm BM}\in[0,1]$ is a beam-management update-rate factor that accounts for the fact that SSB-based beam management and CSI-RS-based CSI acquisition are updated on different timescales\footnote{SSB-based beam management and CSI-RS-based CSI acquisition are configured on different timescales in NR systems. SSB-RSRP mainly reflects beam-domain power associated with large-scale geometry, path loss, and blockage, which varies more slowly than the phase-coherent CSI needed for beamforming. Therefore, RSRP fingerprints can be updated less frequently and reused across multiple CSI-RS acquisition intervals~\cite{Giordani2019BeamManagementNR,3GPP38214,3GPP38215}.}. Therefore, within the NR limited-feedback pipeline, supposing that one coherence interval contains $T_c$ channel uses, problem~\eqref{eq:opt_prob} is reformulated as the following overhead-aware effective-rate maximization problem:
\begin{subequations}
	\begin{align}
		\max_{\mathcal{K},\;\mathbf{W}}\quad &\left(1 - \frac{T_o}{T_c}\right) \sum_{k \in \mathcal{K}} R_k(\mathbf{W}), \\  
		\mathrm{s.t.}\quad &\|\mathbf{w}_k\| = 1,\ \mathcal{K} \subseteq \mathcal{K}_{\rm full}, |\mathcal{K}|=K.
		\label{eq:ref_prob}
	\end{align}
\end{subequations}
For the conventional Type-II pipeline, the SSB stage sends $B$ SSB beams and requires each candidate user reports the index and RSRP of its strongest SSB beam. The CSI-RS stage transmits $N_t$ pilot dimensions to all candidate users and collects $Q$ index and coefficient pairs from these users. Therefore, the online overhead of the conventional pipeline is given by
\begin{equation}
	T_o^{\rm Conv} = \zeta_{\rm BM}(B+2|\mathcal{K}_{\rm full}|) + N_t + 2Q|\mathcal{K}_{\rm full}|.
	\label{eq:overhead_full}
\end{equation}

\vspace{-0.5cm}
\subsection{Proposed Site-Specific Framework}
The limitations of the conventional pipeline arise mainly from two aspects. Firstly, the discovery of Type-II subspace imposes heavy processing burden on the user. Second, equation~\eqref{eq:overhead_full} reveals that the CSI-RS overhead scales with the number of BS antennas $N_t$ and the size of the candidate pool $|\mathcal{K}_{\rm full}|$, which would be prohibitive with growing $N_t$ and $|\mathcal{K}_{\rm full}|$ for larger beamforming gain and more user diversity. These drawbacks motivate the proposed framework that exploits the SSI to reduce the users' computational burden and the online overhead of CSI acquisition and feedback.

\subsubsection{Single-User Case (Prior Work~\cite{prior_singleuser})}
The single-user site-specific Type-II framework in~\cite{prior_singleuser} mainly keeps the limited-feedback structure reviewed above, but changes how the dominant Type-II subspace is obtained. In~\cite{prior_singleuser}, the dominant-subspace discovery is moved before CSI-RS transmission. More specifically, the BS jointly learns a site-specific SSB probing codebook $\mathbf{B}_\theta$ and subspace predictor $\Psi_\theta:\mathbf{r}_k\mapsto\widehat{\mathbf{U}}_k$, with which the SSB RSRP fingerprint is mapped to a predicted basis for the dominant subspace $\mathcal{U}_k$ in \textbf{Assumption~\ref{asm:channel}}, i.e., $\mathrm{span}(\widehat{\mathbf{U}}_k)\approx\mathcal{U}_k$. The BS then sends CSI-RS within $\mathrm{span}(\widehat{\mathbf{U}}_k)$, and the user estimates only the low-dimensional effective CSI inside this predicted subspace. Therefore, the main online burden is shifted from full-dimensional user-side search and feedback to BS-side subspace inference.
\\ \indent
This prior design shows that RSRP fingerprints can be mapped to dominant transmit subspaces before dedicated CSI-RS transmission. In this paper, we use the same learned codebook $\mathbf{B}_\theta$ and mapping $\Psi_\theta$ as the basic subspace predictor and move the focus from single-link CSI acquisition to multi-user grouping, group CSI acquisition, and MU beamforming.

\begin{figure}[t]
	\centering
	\includegraphics[scale=0.26]{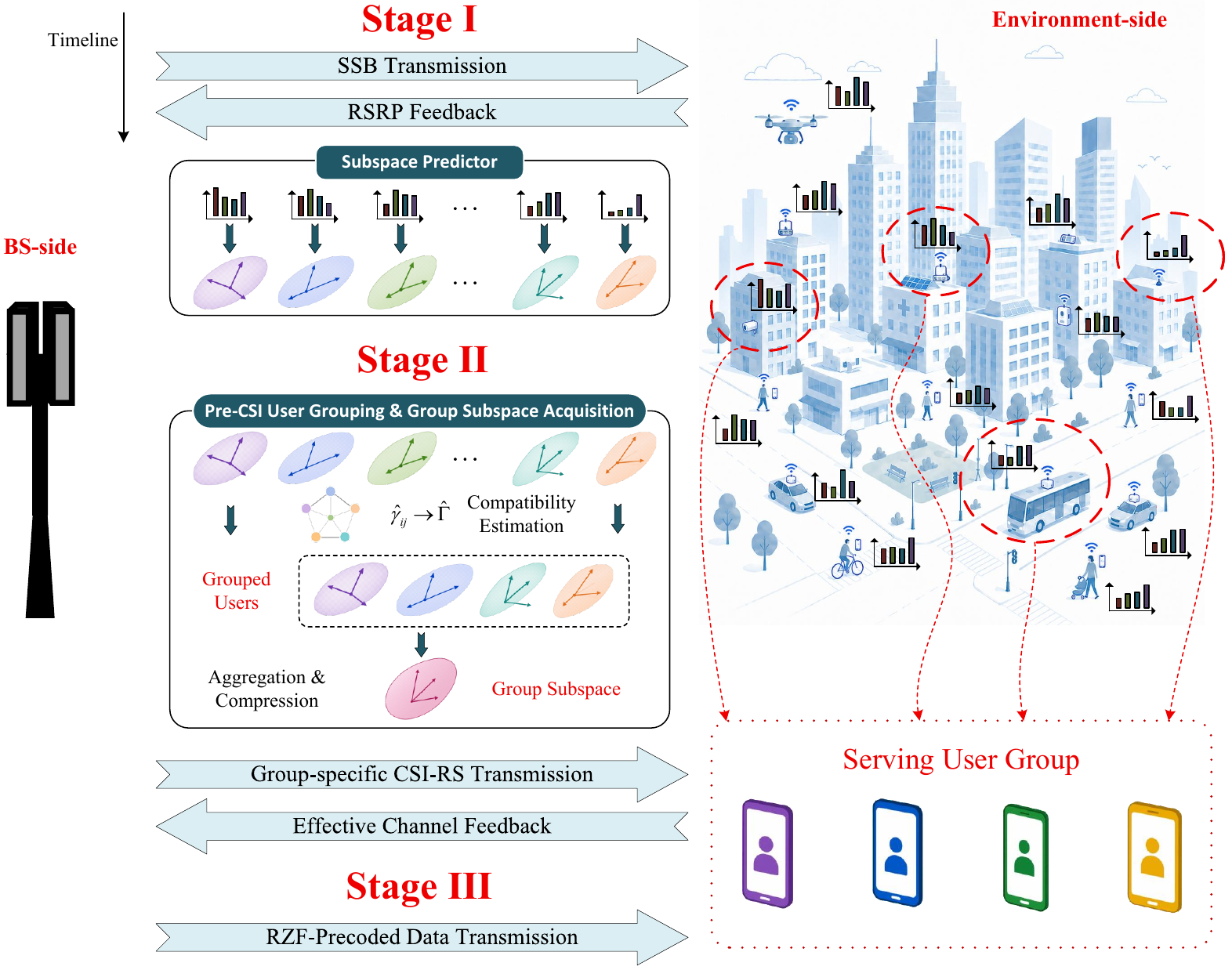}
	\caption{Illustration for the proposed MU beamforming framework.}
	\label{fig:pipeline}
	\vspace{-0.5cm}
\end{figure}

\subsubsection{Multi-User Case}
Although only a subset of users are selected to serve, all users in the candidate pool must feed back their compressed CSI in the conventional pipeline. The proposed framework addresses this bottleneck by exploiting SSI at the BS to screen compatibility and group users before CSI acquisition. Concretely, the proposed site-specific MU beamforming framework consists of the following three stages.
\\ \indent
\textbf{Stage 1 (SSB-based Initial Access):} This stage is identical to the conventional pipeline to establish initial access and collect RSRP fingerprints $\{\mathbf{r}_k\}_{k\in\mathcal{K}_{\rm full}}$ for all candidate users using learned codebook $\mathbf{B}_\theta$. This stage requires $T_{\rm SSB}^{\rm Pro}=B+B|\mathcal{K}_{\rm full}|$ for collecting RSRP from all SSB beams.
\\ \indent
\textbf{Stage 2 (CSI-RS-based CSI Refinement):} Based on the collected RSRP fingerprints, the BS applies the predictor $\Psi_\theta$ to obtain $\{\widehat{\mathbf{U}}_k\}_{k\in\mathcal{K}_{\rm full}}$, evaluates the spatial separability among candidate users, and selects a compatible group $\mathcal{K}\subseteq\mathcal{K}_{\rm full}$ without instantaneous CSI. Then, given the selected group, the BS constructs a shared group acquisition subspace $\mathbf{V}_{\mathcal{K},M}\in\mathbb{C}^{N_t\times M}$, where $M\ll N_t$ is the acquisition dimension configured by the BS, and transmits group-specific CSI-RS over this subspace. Let $\mathbf{S}_{\mathcal{K}}^{\rm CSI}\in\mathbb{C}^{M\times L_{\mathcal{K}}}$ denote the group CSI-RS pilot matrix, where $L_{\mathcal{K}}\ge M$ and $\mathbf{S}_{\mathcal{K}}^{\rm CSI}(\mathbf{S}_{\mathcal{K}}^{\rm CSI})^H=\mathbf{I}_M$. The received signal at user $k$ is given by
\begin{equation}
	\mathbf{y}_{k,\mathcal{K}}^{\rm CSI} = \sqrt{P_{\rm CSI}}\, (\mathbf{S}_{\mathcal{K}}^{\rm CSI})^H \mathbf{V}_{\mathcal{K},M}^H\mathbf{h}_k + \mathbf{n}_{k,\mathcal{K}}^{\rm CSI} \in \mathbb{C}^{L_{\mathcal{K}} \times 1},
	\label{eq:group_csi_received}
\end{equation}
where $\mathbf{n}_{k,\mathcal{K}}^{\rm CSI}\sim\mathcal{CN}(\mathbf{0},\sigma_n^2\mathbf{I}_{L_{\mathcal{K}}})$. Each selected user then estimates and feeds back the $M$-dimensional effective channel from~\eqref{eq:group_csi_received}. With the minimum pilot length $L_{\mathcal{K}}=M$, the corresponding overhead in this stage is $T_{\rm CSI\text{-}RS}^{\rm Pro}=M+MK$. Detailed analysis regarding pre-CSI MU compatibility inference and dimension-constrained group CSI acquisition will be developed in Sections~\ref{sec:compatibility} and~\ref{sec:group_csi}.
\\ \indent
\textbf{Stage 3 (Low-Dimensional MU Beamforming):} With the effective channel feedback, the BS applies the same RZF rule as in the conventional pipeline, but over the low-dimensional effective channel in the group subspace. The reduced-domain RZF beamforming matrix is then mapped back to the antenna domain through $\mathbf{V}_{\mathcal{K},M}$, yielding the final MU beamforming matrix $\mathbf{W}$ for downlink transmission.
\\ \indent
In this way, the overhead of the proposed framework is given by 
\begin{equation}
	T_o^{\rm Pro} = \zeta_{\rm BM}(B+B|\mathcal{K}_{\rm full}|) + M + MK.
	\label{eq:overhead_pro}
\end{equation}
The proposed framework thus reduces overhead on two axes simultaneously: the pilot dimension shrinks from $N_t$ to $M\ll N_t$, and the number of reporting users shrinks from $|\mathcal{K}_{\rm full}|$ to $K\ll|\mathcal{K}_{\rm full}|$. Fig.~\ref{fig:pipeline} gives an illustration for the proposed MU downlink transmission framework.

\section{Pre-CSI MU Compatibility Inference and Group Selection}
\label{sec:compatibility}
Following the proposed framework, this section develops the pre-CSI grouping part for the \textbf{Stage 2} before CSI-RS transmission. Reusing the single-user predictor $\Psi_\theta$ developed in Section VI of~\cite{prior_singleuser} to obtain each $\widehat{\mathbf{U}}_k=\Psi_\theta(\mathbf{r}_k)$ from the RSRP fingerprint, the BS quantifies cross-user separability from the mutual overlap of the predicted subspaces and selects the group on this pre-CSI surrogate.

\subsection{Pairwise Compatibility}
For a pair of candidate users, their compatibility for MU transmission is mainly determined by their \emph{spatial separability}, i.e., whether their channels can be separated by a linear MU beamforming matrix with limited residual IUI~\cite{Yoo2006SUS,Adhikary2013JSDM}. For beamforming methods such as RZF, two users are favorable for co-scheduling when their channel directions are sufficiently separated. Otherwise, interference suppression requires large power loss or noise enhancement. Conventional CSI-based schedulers therefore evaluate pairwise separability from the instantaneous normalized channel overlap using CSI:
\begin{equation}
	\gamma_{ij}^{\rm CSI} \triangleq \frac{|\mathbf{h}_i^H\mathbf{h}_j|^2} {\|\mathbf{h}_i\|^2\|\mathbf{h}_j\|^2}.
	\label{eq:csi_pair_overlap}
\end{equation}
A small $\gamma_{ij}^{\rm CSI}$ hence indicates better spatial separability.
\\ \indent
In the proposed pre-CSI framework, $\gamma_{ij}^{\rm CSI}$ is unavailable. Under \textbf{Assumption~\ref{asm:channel}}, however, each channel $\mathbf{h}_k=\mathbf{U}_k\mathbf{g}_k$ consists of a slowly varying transmit subspace $\mathbf{U}_k$ and fast-fading coefficients $\mathbf{g}_k$. The latter are unknown before CSI-RS transmission, whereas the former can be approximated by the RSRP-predicted subspace. This motivates using the mutual overlap between the users' transmit subspaces as a fading-independent surrogate for pairwise separability. For any two user subspaces, we define the normalized subspace overlap as
\begin{equation}
	\gamma_{ij}\triangleq \frac{1}{Q}\bigl\|\mathbf{U}_i^H\mathbf{U}_j\bigr\|_F^2 = \frac{1}{Q}\mathrm{tr}(\mathbf{P}_{\mathcal{U}_i}\mathbf{P}_{\mathcal{U}_j})\in[0,1],
	\label{eq:subspace_overlap}
\end{equation}
where $\mathbf{P}_{\mathcal{U}_k}$ denotes the orthogonal projector onto $\mathcal{U}_k$ introduced in Section~\ref{sec:system}. The next two propositions show how the subspace overlap relates to the expected and worst-case channel overlap, respectively.

\begin{proposition}[\emph{Subspace overlap and expected channel overlap}]
\label{prop:overlap}\normalfont
Under \textbf{Assumption~\ref{asm:channel}}, the expected channel overlap satisfies
\begin{equation}
	\mathbb{E}\bigl[|\mathbf{h}_i^H \mathbf{h}_j|^2\bigr] = \mathrm{tr}\bigl(\mathbf{R}_i\, \mathbf{U}_i^H \mathbf{U}_j\, \mathbf{R}_j\, \mathbf{U}_j^H \mathbf{U}_i\bigr).
	\label{eq:overlap_general}
\end{equation}
If the subspace coefficients are isotropic, i.e., $\mathbf{R}_k = (\beta_k/Q)\mathbf{I}_Q$, we have
\begin{equation}
	\mathbb{E}\bigl[|\mathbf{h}_i^H \mathbf{h}_j|^2\bigr] = \beta_i\beta_jQ^{-1}\gamma_{ij}.
	\label{eq:overlap_isotropic}
\end{equation}
Correspondingly, the normalized CSI overlap satisfies
\begin{equation}
	\mathbb{E}\bigl[\gamma_{ij}^{\rm CSI}\bigr] = \frac{1}{Q^2}\bigl\|\mathbf{U}_i^H\mathbf{U}_j\bigr\|_F^2 = \frac{1}{Q}\gamma_{ij}.
	\label{eq:gamma_csi_expectation}
\end{equation}
\end{proposition}
\begin{IEEEproof}
	Substituting $\mathbf{h}_k=\mathbf{U}_k\mathbf{g}_k$ into $\mathbb{E}[|\mathbf{h}_i^H\mathbf{h}_j|^2]$ reduces the problem to averaging over the independent subspace coefficients $\mathbf{g}_i$ and $\mathbf{g}_j$. Conditioning on $\mathbf{g}_i$, the expectation over $\mathbf{g}_j$ replaces $\mathbf{g}_j\mathbf{g}_j^H$ by $\mathbf{R}_j$. The remaining term is a quadratic form in $\mathbf{g}_i$; taking its expectation and using $\mathbb{E}[\mathbf{g}_i\mathbf{g}_i^H]=\mathbf{R}_i$ gives the trace expression in~\eqref{eq:overlap_general}. If $\mathbf{R}_k=(\beta_k/Q)\mathbf{I}_Q$, then~\eqref{eq:overlap_general} becomes $\frac{\beta_i\beta_j}{Q^2}\|\mathbf{U}_i^H\mathbf{U}_j\|_F^2=\beta_i\beta_jQ^{-1}\gamma_{ij}$, which proves~\eqref{eq:overlap_isotropic}. For the normalized CSI overlap, $\|\mathbf{h}_k\|=\|\mathbf{g}_k\|$ because $\mathbf{U}_k$ has orthonormal columns. Under the isotropic model, the normalized direction satisfies $\mathbb{E}[\mathbf{g}_k\mathbf{g}_k^H/\|\mathbf{g}_k\|^2]=Q^{-1}\mathbf{I}_Q$. Applying the same argument to these normalized directions yields~\eqref{eq:gamma_csi_expectation}.
\end{IEEEproof}

\textbf{Proposition~\ref{prop:overlap}} shows that the expected channel overlap is governed by the relative geometry of the two transmit subspaces, with its exact value further weighted by the coefficient covariances $\mathbf{R}_i$ and $\mathbf{R}_j$. The isotropic model removes these mode-dependent weights by spreading the power uniformly over the $Q$ modes of each subspace, and thus yields the clean relation $\mathbb{E}[\gamma_{ij}^{\rm CSI}]=Q^{-1}\gamma_{ij}$ in~\eqref{eq:gamma_csi_expectation}. This assumption is idealized, since practical channels generally have unequal mode powers, and exact numerical agreement between $\gamma_{ij}$ and $\gamma_{ij}^{\rm CSI}$ should not be expected in all channel realizations. Nevertheless, the general expression~\eqref{eq:overlap_general} still shows that the common subspace geometry is the structural source of channel coupling, while the unknown mode powers only reweight the shared directions. Since these powers and the instantaneous fading coefficients are unavailable before CSI acquisition, the normalized subspace overlap $\gamma_{ij}$ still provides a power-agnostic geometric reference for assessing potential inter-user coupling. Its practical role is therefore to preserve the compatibility ordering rather than to predict the exact CSI-domain overlap value, which is verified empirically in the first part of Section~\ref{sec:results_stage1}.

\begin{proposition}[\emph{Subspace overlap and worst-case channel overlap}]
\label{prop:worst_case}\normalfont
Under \textbf{Assumption~\ref{asm:channel}}, the worst-case channel overlap over all channel realizations is given by
\begin{equation}
	\sup_{\mathbf{g}_i,\mathbf{g}_j \neq \mathbf{0}} \frac{|\mathbf{h}_i^H \mathbf{h}_j|^2} {\|\mathbf{h}_i\|^2 \|\mathbf{h}_j\|^2}	= \bigl\|\mathbf{U}_i^H \mathbf{U}_j\bigr\|_2^2	= \bigl\|\mathbf{P}_{\mathcal{U}_i}\mathbf{P}_{\mathcal{U}_j}\bigr\|_2^2 = \rho_{ij},
	\label{eq:worst_case_overlap}
\end{equation}
where the supremum is attained when $\mathbf{g}_i$ and $\mathbf{g}_j$ are aligned to the dominant singular vectors of $\mathbf{U}_i^H \mathbf{U}_j$.
\end{proposition}

\begin{IEEEproof}
By applying the Cauchy–Schwarz inequality, we obtain
\begin{equation}
	|\mathbf{h}_i^H \mathbf{h}_j| = |\mathbf{g}_i^H \mathbf{U}_i^H\mathbf{U}_j \mathbf{g}_j| \leq \|\mathbf{U}_i^H\mathbf{U}_j\|_2\,\|\mathbf{g}_i\|\|\mathbf{g}_j\|,
	\label{26}
\end{equation}
where equality holds when $\mathbf{g}_i$ and $\mathbf{g}_j$ align with the dominant left and right singular vectors of $\mathbf{U}_i^H\mathbf{U}_j$, respectively. Squaring both sides of~\eqref{26}. and exploiting $\|\mathbf{h}_k\|=\|\mathbf{g}_k\|$ gives the first equality in~\eqref{eq:worst_case_overlap}. The projector form follows from $\mathbf{P}_i\mathbf{P}_j=\mathbf{U}_i(\mathbf{U}_i^H\mathbf{U}_j)\mathbf{U}_j^H$, whose nonzero singular values are those of $\mathbf{U}_i^H\mathbf{U}_j$.
\end{IEEEproof}

\textbf{Proposition~\ref{prop:worst_case}} complements \textbf{Proposition~\ref{prop:overlap}} by characterizing the strongest possible channel alignment within two user subspaces. If $\{\sigma_{ij,q}\}_{q=1}^{Q}$ are the singular values of $\mathbf{U}_i^H\mathbf{U}_j$, the average and strongest subspace overlaps are represented by
\begin{equation}
	\rho_{ij}=\max_q \sigma_{ij,q}^2,\qquad \gamma_{ij}=\frac{1}{Q}\sum_{q=1}^{Q}\sigma_{ij,q}^2.
	\label{eq:rho_gamma_relation}
\end{equation}
Thus, $\gamma_{ij}$ measures the average overlap across the shared principal modes, whereas $\rho_{ij}$ measures the largest shared mode and captures the worst-case IUI risk. Since $\gamma_{ij}\leq\rho_{ij}\leq Q\gamma_{ij}$, reducing $\gamma_{ij}$ also limits the possible growth of the strongest overlap when $Q$ is fixed, although it is less conservative than directly minimizing $\rho_{ij}$. We therefore adopt $\gamma_{ij}$ as the main compatibility metric in this paper.
\\ \indent 
However, these quantities still depend on the ground-truth subspaces $\mathbf{U}_k$, which are also unavailable before CSI acquisition. The BS therefore substitutes the RSRP-predicted orthonormal basis $\widehat{\mathbf{U}}_k$ and evaluates the corresponding \emph{predicted compatibility proxies}:
\begin{equation}
	\widehat{\gamma}_{ij} = \frac{1}{Q}\bigl\|\widehat{\mathbf{U}}_i^H \widehat{\mathbf{U}}_j\bigr\|_F^2, \qquad \widehat{\rho}_{ij} = \bigl\|\widehat{\mathbf{U}}_i^H \widehat{\mathbf{U}}_j\bigr\|_2^2.
	\label{eq:gamma_hat}
\end{equation}
The following proposition bounds their deviations from the true subspace-overlap quantities.

\begin{proposition}[\emph{Compatibility proxy accuracy}]
\label{prop:proxy_accuracy}\normalfont
Define the per-user \emph{subspace alignment quality}
\begin{equation}
	\xi_k \triangleq \frac{1}{Q}\bigl\|\widehat{\mathbf{U}}_k^H \mathbf{U}_k\bigr\|_F^2 \in [0,1].
	\label{eq:xi_alignment}
\end{equation}
Then the \emph{predicted compatibility proxies} satisfy
\begin{subequations}
	\begin{align}
		\bigl|\widehat{\gamma}_{ij} - \gamma_{ij}\bigr| &\leq \sqrt{2} \bigl(\sqrt{1-\xi_i} + \sqrt{1-\xi_j}\bigr), \label{eq:proxy_error}\\
		\bigl|\sqrt{\widehat{\rho}_{ij}}-\sqrt{\rho_{ij}}\bigr| &\leq \sqrt{Q(1-\xi_i)}+\sqrt{Q(1-\xi_j)}.
		\label{eq:rho_sqrt_error}
	\end{align}
\end{subequations}
\end{proposition}

\begin{IEEEproof}
Let $\widehat{\mathbf{P}}_k=\widehat{\mathbf{U}}_k\widehat{\mathbf{U}}_k^H$. From~\eqref{eq:subspace_overlap} and~\eqref{eq:gamma_hat}, $\widehat{\gamma}_{ij}-\gamma_{ij}$ is the normalized trace difference between $\widehat{\mathbf{P}}_i\widehat{\mathbf{P}}_j$ and $\mathbf{P}_{\mathcal{U}_i}\mathbf{P}_{\mathcal{U}_j}$. To separate the two users' subspace-prediction errors, we add and subtract $\mathbf{P}_{\mathcal{U}_i}\widehat{\mathbf{P}}_j$ inside the trace. The triangle inequality and $|\mathrm{tr}(\mathbf{A}\mathbf{B})| \leq \|\mathbf{A}\|_F\|\mathbf{B}\|_F$ then give
\begin{equation}
	\bigl|\widehat{\gamma}_{ij}-\gamma_{ij}\bigr| \leq \frac{1}{Q}\bigl( \|\widehat{\mathbf{P}}_i-\mathbf{P}_{\mathcal{U}_i}\|_F\|\widehat{\mathbf{P}}_j\|_F+\|\mathbf{P}_{\mathcal{U}_i}\|_F\|\widehat{\mathbf{P}}_j-\mathbf{P}_{\mathcal{U}_j}\|_F \bigr).
	\label{eq:proxy_mid}
\end{equation}
Moreover, since $\|\mathbf{P}_{\mathcal{U}_k}\|_F = \|\widehat{\mathbf{P}}_k\|_F = \sqrt{Q}$, we have
\begin{align}
	\|\widehat{\mathbf{P}}_k - \mathbf{P}_{\mathcal{U}_k}\|_F^2 &= \mathrm{tr}\bigl((\widehat{\mathbf{P}}_k-\mathbf{P}_{\mathcal{U}_k})^2\bigr) \nonumber= 2Q-2\mathrm{tr}(\widehat{\mathbf{P}}_k\mathbf{P}_{\mathcal{U}_k}) \\ &=2Q-2\|\widehat{\mathbf{U}}_k^H\mathbf{U}_k\|_F^2 \nonumber =2Q(1-\xi_k).
\end{align}
Together with~\eqref{eq:proxy_mid}, this gives
\begin{align}
	\bigl|\widehat{\gamma}_{ij} - \gamma_{ij}\bigr| &\leq \frac{1}{Q}\bigl(\sqrt{2Q(1-\xi_i)}\sqrt{Q} +\sqrt{Q}\sqrt{2Q(1-\xi_j)}\bigr) \nonumber\\ &= \sqrt{2}\bigl(\sqrt{1-\xi_i}+\sqrt{1-\xi_j}\bigr),
\end{align}
which proves~\eqref{eq:proxy_error}. For $\rho_{ij}$, since $\sqrt{\rho_{ij}}=\|\mathbf{P}_{\mathcal{U}_i}\mathbf{P}_{\mathcal{U}_j}\|_2$ and $\sqrt{\widehat{\rho}_{ij}}=\|\widehat{\mathbf{P}}_i\widehat{\mathbf{P}}_j\|_2$, the reverse triangle inequality leads to
\begin{align}
	\bigl|\sqrt{\widehat{\rho}_{ij}}-\sqrt{\rho_{ij}}\bigr| &\leq 	\|\widehat{\mathbf{P}}_i\widehat{\mathbf{P}}_j-\mathbf{P}_{\mathcal{U}_i}\mathbf{P}_{\mathcal{U}_j}\|_2 \nonumber\\ &\leq \|\widehat{\mathbf{P}}_i-\mathbf{P}_{\mathcal{U}_i}\|_2 +\|\widehat{\mathbf{P}}_j-\mathbf{P}_{\mathcal{U}_j}\|_2.
\end{align}
For the principal angles $\{\theta_{k,q}\}_{q=1}^{Q}$ between $\widehat{\mathcal{U}}_k$ and $\mathcal{U}_k$, we have
\begin{equation}
	\|\widehat{\mathbf{P}}_k-\mathbf{P}_{\mathcal{U}_k}\|_2^2 = \max_q \sin^2\theta_{k,q} \leq \sum_{q=1}^{Q}\sin^2\theta_{k,q} = Q(1-\xi_k),
\end{equation}
which proves~\eqref{eq:rho_sqrt_error}.
\end{IEEEproof}

\textbf{Proposition~\ref{prop:proxy_accuracy}} quantifies how subspace-prediction error propagates to the pairwise compatibility proxies. The bounds decrease monotonically with $\xi_k$ and vanish as $\xi_k \to 1$, indicating that prediction errors cannot grow arbitrarily into compatibility errors. Since these are worst-case bounds and may be conservative, their practical tightness is evaluated empirically through the pair-level results in Section~\ref{sec:results_stage1}. Note that $\xi_k$ is only an analysis and validation metric for evaluating prediction accuracy and is not required by the online scheduler.
\vspace{-1cm}
\subsection{Greedy Group Selection}
The preceding subsection provides pairwise pre-CSI compatibility proxies for all candidate users in $\mathcal{K}_{\rm full}$. To form a MU group, these pairwise quantities must be converted into a group-level selection criterion. For linear MU beamforming, a single strongly coupled user pair can dominate residual IUI and degrade the conditioning of the effective MU channel, even when the other pairs in the group are well separated~\cite{Yoo2006SUS}. We therefore characterize a candidate group by its largest pairwise overlap. For any group $\mathcal{G}\subseteq\mathcal{K}_{\rm full}$, define the true and predicted group overlaps as
\begin{equation}
	\Gamma(\mathcal{G}) \triangleq \max_{i<j,\;i,j\in\mathcal{G}}\gamma_{ij},\qquad \widehat{\Gamma}(\mathcal{G}) \triangleq \max_{i<j,\;i,j\in\mathcal{G}}\widehat{\gamma}_{ij}.
	\label{eq:group_overlap_def}
\end{equation}
Since scheduling uses $\widehat{\Gamma}(\mathcal{G})$ while the desired compatibility criterion is $\Gamma(\mathcal{G})$, the next step is to verify that the pairwise prediction accuracy remains controlled after the max aggregation. The following proposition gives this group-level error bound.

\begin{proposition}[\emph{Group overlap accuracy}]
\label{prop:group_coherence}\normalfont
If $\xi_{\min}\triangleq\min_{k\in\mathcal{G}}\xi_k$ and $|\mathcal{G}|\geq2$, then we have
\begin{equation}
	\bigl|\Gamma(\mathcal{G})-\widehat{\Gamma}(\mathcal{G})\bigr| \leq 2\sqrt{2(1-\xi_{\min})}.
	\label{eq:group_coherence_bound}
\end{equation}
\end{proposition}

\begin{IEEEproof}
By \emph{Proposition~\ref{prop:proxy_accuracy}}, every pair $(i,j)$ in $\mathcal{G}$ satisfies
\begin{equation}
	\bigl|\widehat{\gamma}_{ij}-\gamma_{ij}\bigr| \leq \sqrt{2}\bigl(\sqrt{1-\xi_i}+\sqrt{1-\xi_j}\bigr) \leq 2\sqrt{2(1-\xi_{\min})}.
\end{equation}
Thus, $\gamma_{ij}\leq \widehat{\gamma}_{ij} +2\sqrt{2(1-\xi_{\min})}$ and $\widehat{\gamma}_{ij}\leq \gamma_{ij}+2\sqrt{2(1-\xi_{\min})}$ hold for all pairs in $\mathcal{G}$. Taking the maximum over all pairs in both inequalities gives~\eqref{eq:group_coherence_bound}.
\end{IEEEproof}

\textbf{Proposition~\ref{prop:group_coherence}} extends this bounded-error relation from pairwise overlap to the group-level bottleneck metric. When the per-user subspace predictions are reliable, the predicted group bottleneck remains close to the true subspace bottleneck, while the practical selection behavior is evaluated in the second part of Section~\ref{sec:results_stage1}. The scheduler therefore seeks a size-$K$ group that minimizes $\widehat{\Gamma}(\mathcal{G})$:
\begin{equation}
	\mathcal{K} = \operatorname*{arg\,min}_{\mathcal{G}\subseteq\mathcal{K}_{\rm full},\,|\mathcal{G}|=K} \widehat{\Gamma}(\mathcal{G}).
	\label{eq:pre_csi_grouping}
\end{equation}
Solving~\eqref{eq:pre_csi_grouping} exactly requires checking $\binom{|\mathcal{K}_{\rm full}|}{K}$ candidate groups, which is combinatorial and becomes prohibitive for a large candidate pool. We therefore adopt a low-complexity greedy rule inspired by SUS~\cite{Yoo2006SUS}, but move the selection metric from the instantaneous CSI domain to the predicted subspace domain. We refer to this rule as \emph{pre-CSI subspace-SUS}. Since the overlap metric measures separability but not link strength, the online scheduler first selects the user with the largest estimated large-scale channel power:
\begin{equation}
	k_1 = \operatorname*{arg\,max}_{k\in\mathcal{K}_{\rm full}}\widehat{\beta}_k,\quad \mathcal{K}=\{k_1\},
	\label{eq:greedy_seed}
\end{equation}
where $\widehat{\beta}_k$ is estimated from the SSB RSRP fingerprint $\mathbf{r}_k$\footnote{The channel power $\beta_k=\mathrm{tr}(\mathbf{R}_k)$ is estimated from the SSB RSRP fingerprint $\mathbf{r}_k$. As a standardized power measurement, RSRP reflects the large-scale received power and is routinely used for large-scale gain and pathloss estimation~\cite{3GPP38215,Giordani2019BeamManagementNR}.}. The BS then adds the user whose largest predicted overlap with the current group is the smallest, i.e., it iteratively performs:
\begin{equation}
	k^\star = \operatorname*{arg\,min}_{k \in \mathcal{K}_{\rm full}\setminus\mathcal{K}} \max_{j \in \mathcal{K}}\widehat{\gamma}_{kj},
	\label{eq:greedy}
\end{equation}
until $|\mathcal{K}| = K$. Equivalently, minimizing the maximum predicted overlap is the same as maximizing the minimum separation among the selected predicted subspaces, since $\widehat{\gamma}_{ij}$ decreases monotonically with the normalized chordal distance between users $i$ and $j$. The update in~\eqref{eq:greedy} therefore implements a farthest-first traversal in the predicted subspace domain, consistent with classical greedy selection methods~\cite{Gonzalez1985FarthestFirst}. Combined with the overlap-based compatibility surrogate, it turns the pre-CSI grouping problem into a sequential selection procedure: the BS only updates the worst predicted overlap between each remaining candidate and the current group. With precomputed pairwise overlaps, the proposed pre-CSI subspace-SUS reduces the search complexity from combinatorial in $|\mathcal{K}_{\rm full}|$ to $\mathcal{O}(K|\mathcal{K}_{\rm full}|)$.
\vspace{-0.3cm}
\section{Low-Dimensional Group CSI Acquisition and Beamforming}
\label{sec:group_csi}
This section develops the group CSI acquisition part of \textbf{Stage 2}. Given the selected group $\mathcal{K}$ and its predicted subspaces $\{\widehat{\mathbf{U}}_k\}_{k\in\mathcal{K}}$, the BS constructs a compact group CSI-RS subspace $\mathcal{V}_{\mathcal{K},M}$ with basis $\mathbf{V}_{\mathcal{K},M}\in\mathbb{C}^{N_t\times M}$, over which the selected users feed back low-dimensional effective channels for subsequent RZF beamforming. We then analyze how this subspace preserves channel energy and effective-channel separability.
\vspace{-0.3cm}
\subsection{Group-Specific CSI-RS Subspace Design}
For the selected group $\mathcal{K}$, reducing the overall union-subspace dimension requires the BS to identify a compact set of CSI-RS probing directions that preserves the dominant group-energy components of the scheduled users, rather than probing the full concatenation of all predicted subspaces. A natural compression criterion is to maximize the total captured channel energy:
\begin{equation}
	\max_{\mathbf{V} \in \mathbb{C}^{N_t \times M}}\  \sum_{k \in \mathcal{K}} \mathbb{E}\!\left[\|\mathbf{V}^H\mathbf{h}_k\|^2\right] \quad \mathrm{s.t.}\  \mathbf{V}^H\mathbf{V}=\mathbf{I}_M.
	\label{eq:oracle_group_capture_opt}
\end{equation}
Problem~\eqref{eq:oracle_group_capture_opt} is intractable before CSI acquisition because it depends on unavailable channel statistics. The following proposition gives a tractable characterization of its solution under the isotropic coefficient model.

\begin{proposition}[\emph{Capture optimality of group CSI-RS subspace}]
\label{prop:svd_opt}\normalfont
Under \textbf{Assumption~\ref{asm:channel}} with isotropic $\mathbf{R}_k = (\beta_k/Q)\mathbf{I}_Q$, the expected captured energy of user $k$ over any $M$-dimensional subspace with orthonormal basis $\mathbf{V}$ satisfies
\begin{equation}
	\mathbb{E}\!\left[\|\mathbf{V}^H\mathbf{h}_k\|^2\right] = \frac{\beta_k}{Q}\bigl\|\mathbf{V}^H\mathbf{U}_k\bigr\|_F^2.
	\label{eq:expected_capture_formula}
\end{equation}
Let $\mathbf{R}_{\mathcal{K}}^\star = \sum_{k \in \mathcal{K}}\omega_k^\star\mathbf{U}_k\mathbf{U}_k^H$ denote the group spatial covariance, i.e., the power-weighted sum of the scheduled users' transmit-subspace projectors with $\omega_k^\star = \beta_k/\sum_{\ell\in\mathcal{K}}\beta_\ell$. Then the optimizer of~\eqref{eq:oracle_group_capture_opt} is given by
\begin{equation}
	\mathbf{V}_{\mathcal{K},M}^{\star} = \mathrm{TopEig}_M(\mathbf{R}_{\mathcal{K}}^\star),
	\label{eq:oracle_group_subspace}
\end{equation}
where $\mathrm{TopEig}_M(\cdot)$ denotes an orthonormal basis of the dominant $M$-dimensional eigenspace of a Hermitian matrix.
\end{proposition}

\begin{IEEEproof}
\textbf{Assumption~\ref{asm:channel}} and the isotropic model leads to
\begin{align}
	\mathbb{E}\!\left[\|\mathbf{V}^H\mathbf{h}_k\|^2\right] &= \mathbb{E}\bigl[\|\mathbf{V}^H\mathbf{U}_k\mathbf{g}_k\|^2\bigr] = \mathrm{tr}\!\bigl(\mathbf{V}^H\mathbf{U}_k\mathbf{R}_k \mathbf{U}_k^H\mathbf{V}\bigr) \notag \\ &= \frac{\beta_k}{Q}\bigl\|\mathbf{V}^H\mathbf{U}_k\bigr\|_F^2,
\end{align}
establishing~\eqref{eq:expected_capture_formula}. Summing over users yields
\begin{equation}
	\sum_{k \in \mathcal{K}} \mathbb{E}\!\left[\|\mathbf{V}^H\mathbf{h}_k\|^2\right] = \frac{\sum_{\ell\in\mathcal{K}}\beta_\ell}{Q} \mathrm{tr}\!\bigl(\mathbf{V}^H \mathbf{R}_{\mathcal{K}}^\star \mathbf{V}\bigr).
\end{equation}
By the Ky Fan maximum principle~\cite{HornJohnson2012MatrixAnalysis}, the Stiefel maximizer of $\mathrm{tr}(\mathbf{V}^H\mathbf{A}\mathbf{V})$ subject to $\mathbf{V}^H\mathbf{V} = \mathbf{I}_M$ consists of the top-$M$ eigenvectors of $\mathbf{A}$ and is unique when $\mathbf{A}$ has a spectral gap at position $M$.
\end{IEEEproof}

\textbf{Proposition~\ref{prop:svd_opt}} motivates an implementable surrogate for problem~\eqref{eq:oracle_group_capture_opt} by replacing the unavailable $\mathbf{U}_k$ and $\beta_k$ with the RSRP-predicted $\widehat{\mathbf{U}}_k$ and $\widehat{\beta}_k$. With $\omega_k=\widehat{\beta}_k/\sum_{\ell\in\mathcal{K}}\widehat{\beta}_\ell$, the predicted capture problem becomes
\begin{equation}
	\max_{\mathbf{V}^H\mathbf{V}=\mathbf{I}_M} \sum_{k \in \mathcal{K}} \omega_k \bigl\|\mathbf{V}^H\widehat{\mathbf{U}}_k\bigr\|_F^2.
	\label{eq:pred_group_capture_opt}
\end{equation}
By defining the corresponding \emph{predicted group spatial covariance}
\begin{equation}
	\widehat{\mathbf{R}}_{\mathcal{K}} = \sum_{k \in \mathcal{K}} \omega_k\, \widehat{\mathbf{U}}_k \widehat{\mathbf{U}}_k^H,
	\label{eq:group_profile}
\end{equation}
the solution of problem~\eqref{eq:pred_group_capture_opt} is the dominant $M$-dimensional eigenspace of $\widehat{\mathbf{R}}_{\mathcal{K}}$, i.e.,
\begin{equation}
	\mathbf{V}_{\mathcal{K},M} = \mathrm{TopEig}_M\!\bigl(\widehat{\mathbf{R}}_{\mathcal{K}}\bigr) \in \mathbb{C}^{N_t \times M},
	\label{eq:group_subspace}
\end{equation}
whose columns form an orthonormal CSI-RS acquisition basis. $\mathcal{V}_{\mathcal{K},M}=\mathrm{span}(\mathbf{V}_{\mathcal{K},M})$ is the resulting group-level acquisition subspace, which preserves the dominant weighted spatial components of the scheduled group.
\\ \indent
Following the \textbf{Stage 2} CSI-RS transmission model in~\eqref{eq:group_csi_received}, each selected user estimates and feeds back the low-dimensional effective channel
\begin{equation}
	\widetilde{\mathbf{h}}_k = \mathbf{V}_{\mathcal{K},M}^H \mathbf{h}_k    \in \mathbb{C}^M,
	\label{eq:effective_channel}
\end{equation}
which contains the coordinates of the channel component captured by $\mathcal{V}_{\mathcal{K},M}$. Equivalently, the above equation can be read as
\begin{equation}
	\mathbf{V}_{\mathcal{K},M}\widetilde{\mathbf{h}}_k = \mathbf{V}_{\mathcal{K},M}\mathbf{V}_{\mathcal{K},M}^H\mathbf{h}_k = \mathbf{P}_{\mathcal{V}_{\mathcal{K},M}}\mathbf{h}_k,
	\label{eq:effective_projection_relation}
\end{equation}
where $\mathbf{P}_{\mathcal{V}_{\mathcal{K},M}}$ is the orthogonal projector onto the group acquisition subspace. The BS then collects all effective channels into
\begin{equation}
	\widetilde{\mathbf{H}}_{\mathcal{K}} = \begin{bmatrix}
		\widetilde{\mathbf{h}}_{k_1}^H \\ \vdots \\ \widetilde{\mathbf{h}}_{k_{K}}^H
	\end{bmatrix}
	\in \mathbb{C}^{K \times M}.
	\label{eq:Heff}
\end{equation}

\subsection{Low-Dimensional RZF Beamforming and Conditioning}
\label{sec:rzf}
After the BS obtains the effective MU channel $\widetilde{\mathbf{H}}_{\mathcal{K}}$, it computes the MU beamforming matrix from the $M$-dimensional effective channel induced by the group acquisition subspace $\mathcal{V}_{\mathcal{K},M}$. Following the standard RZF form~\cite{Peel2005RZF}, the reduced-domain beamforming matrix is
\begin{equation}
	\widetilde{\mathbf{W}}_{\rm RZF} = \widetilde{\mathbf{H}}_{\mathcal{K}}^H \Bigl(\widetilde{\mathbf{H}}_{\mathcal{K}} \widetilde{\mathbf{H}}_{\mathcal{K}}^H + \frac{K\sigma_n^2}{P_t}\mathbf{I}_K\Bigr)^{-1} \in \mathbb{C}^{M \times K}.
	\label{eq:rzf_lowdim}
\end{equation}
The antenna-domain beamforming vector for user $k$ is obtained by mapping the $k$-th reduced-domain beamforming vector back through the group acquisition basis with unit-norm normalization:
\begin{equation}
	\mathbf{w}_k = \frac{\mathbf{V}_{\mathcal{K},M}\, \widetilde{\mathbf{w}}_k} {\|\mathbf{V}_{\mathcal{K},M}\, \widetilde{\mathbf{w}}_k\|} = \frac{\mathbf{V}_{\mathcal{K},M}\, \widetilde{\mathbf{w}}_k} {\| \widetilde{\mathbf{w}}_k\|},
	\label{eq:W_full}
\end{equation}
where $\widetilde{\mathbf{w}}_k=[\widetilde{\mathbf{W}}_{\rm RZF}]_{:,k}$ and the second equality follows from the orthonormality of $\mathbf{V}_{\mathcal{K},M}$.

\begin{remark}[\emph{RZF beamforming within the group subspace}]\normalfont
	The reduced-domain RZF in~\eqref{eq:rzf_lowdim} is equivalent to applying RZF to the antenna-domain channel restricted to $\mathcal{V}_{\mathcal{K},M}$. In other words, for any reduced-domain beamforming vector $\widetilde{\mathbf{w}}_j$, the coupling from user $j$'s stream to user $i$ is
	\begin{equation}
		\mathbf{h}_i^H\mathbf{V}_{\mathcal{K},M}\widetilde{\mathbf{w}}_j = \bigl(\mathbf{V}_{\mathcal{K},M}^H\mathbf{h}_i\bigr)^H\widetilde{\mathbf{w}}_j = \widetilde{\mathbf{h}}_i^H\widetilde{\mathbf{w}}_j.
		\label{eq:effective_io_equiv}
	\end{equation}
	Thus, the desired signal and IUI are governed by the effective channels once the beamforming vectors are restricted to the group subspace. Computing RZF from $\widetilde{\mathbf{H}}_{\mathcal K}$ therefore preserves the MU input-output relation within the acquired subspace while reducing the matrix computation from dimension $N_t$ to $M$.
\end{remark}

After the RZF beamforming matrix is established, the remaining issue is whether the acquired effective MU channel is sufficiently well conditioned for interference suppression. This depends on two factors: the group subspace must retain enough channel energy for each scheduled user, and the resulting effective channel directions must remain separated. The first factor is quantified by the individual channel capture efficiency, which is defined by
\begin{equation}
	\chi_k = \frac{\|\mathbf{P}_{\mathcal{V}_{\mathcal{K},M}}\mathbf{h}_k\|^2} {\|\mathbf{h}_k\|^2} = \frac{\|\mathbf{V}_{\mathcal{K},M}^H \mathbf{h}_k\|^2} {\|\mathbf{h}_k\|^2} \in [0,1],
	\label{eq:chi_true}
\end{equation}
where $\mathbf{P}_{\mathcal{V}_{\mathcal{K},M}} =\mathbf{V}_{\mathcal{K},M}\mathbf{V}_{\mathcal{K},M}^H$ denotes the orthogonal projector onto the group subspace. It measures the fraction of the channel energy of user $k$ retained by $\mathcal{V}_{\mathcal{K},M}$ and is evaluated in the first part of Section~\ref{sec:results_stage2}. 
\\ \indent 
Given that the captured channel gains are accounted for by $\chi_k$, the second factor is the directional separability of the effective channels. Let $\mathbf{z}_k=\widetilde{\mathbf{h}}_k/\|\widetilde{\mathbf{h}}_k\|$ and stack $\{\mathbf{z}_k^H\}_{k\in\mathcal{K}}$ into $\mathbf{Z}_{\mathcal{K}}\in\mathbb{C}^{K\times M}$. The normalized Gram matrix for the  effective channel matrix is given by
\begin{equation}
	\widetilde{\mathbf{G}} = \mathbf{Z}_{\mathcal{K}}\mathbf{Z}_{\mathcal{K}}^H,
	\label{eq:gram}
\end{equation}
whose diagonal entries are one and off-diagonal entries are the normalized overlaps among effective channels. Hence, a near-identity $\widetilde{\mathbf{G}}$ indicates that the users remain well separated after projection, whereas large off-diagonal entries imply stronger inter-user coupling after CSI acquisition. We summarize this separability by the condition number of $\widetilde{\mathbf{G}}$, i.e.,
\begin{equation}
	\kappa(\widetilde{\mathbf{G}}) = \frac{\lambda_{\max}(\widetilde{\mathbf{G}})}{\lambda_{\min}(\widetilde{\mathbf{G}})}.
	\label{eq:gram_kappa}
\end{equation}
A value close to one indicates nearly orthogonal effective channels, while a large value indicates an ill-conditioned effective MU channel. Its connection to RZF follows from the decomposition
\begin{equation}
	\widetilde{\mathbf{H}}_{\mathcal{K}}\widetilde{\mathbf{H}}_{\mathcal{K}}^H
	= \mathbf{D}_{\mathcal{K}}\widetilde{\mathbf{G}}\mathbf{D}_{\mathcal{K}},
	\label{eq:rzf_gram_decomp}
\end{equation}
where $\mathbf{D}_{\mathcal{K}}=\mathrm{diag}(\|\widetilde{\mathbf{h}}_{k_1}\|,\ldots,\|\widetilde{\mathbf{h}}_{k_K}\|)$ contains the captured effective-channel magnitudes. The matrix inverted by RZF in~\eqref{eq:rzf_lowdim} is therefore $\mathbf{D}_{\mathcal{K}}\widetilde{\mathbf{G}}\mathbf{D}_{\mathcal{K}}+\alpha_{\rm RZF}\mathbf{I}_K$, with $\alpha_{\rm RZF}=K\sigma_n^2/P_t$. Hence, after the per-user captured gains are accounted for by $\chi_k$, $\widetilde{\mathbf{G}}$ determines the directional conditioning of the RZF. In the equal-gain case, the RZF eigenmode gains scale as $(c^2\lambda_i(\widetilde{\mathbf{G}})+\alpha_{\rm RZF})^{-1}$, so a small $\lambda_{\min}(\widetilde{\mathbf{G}})$, or equivalently a large $\kappa(\widetilde{\mathbf{G}})$, indicates stronger noise and CSI-error amplification. This conditioning metric is evaluated empirically in the second part of Section~\ref{sec:results_stage2}. The following theorem further explains when the projection onto the group subspace preserves effective-channel separability.

\begin{theorem}[\emph{Projected-channel overlap bound}]
\label{thm:main}\normalfont
Let $\mathbf{G}_{\mathcal{K}}^{\rm full}$ denote the normalized Gram matrix of the full-dimensional channels, whose $(i,j)$-th entry is $[\mathbf{G}_{\mathcal{K}}^{\rm full}]_{ij} =\mathbf{h}_i^H\mathbf{h}_j/(\|\mathbf{h}_i\|\|\mathbf{h}_j\|)$. Suppose that:
\begin{enumerate}
	\item \emph{Capture:} $\chi_k \geq 1 - \delta$ for all $k \in \mathcal{K}$, and
	\item \emph{Full-dimensional separability:} $|[\mathbf{G}_{\mathcal{K}}^{\rm full}]_{ij}| \leq \epsilon$ for all $i \neq j$, $i,j \in \mathcal{K}$.
\end{enumerate}
Then the off-diagonal entries of the effective Gram matrix satisfy
\begin{equation}
	|[\widetilde{\mathbf{G}}]_{ij}| \leq \frac{\epsilon + \delta}{1-\delta} \;\triangleq\; \tilde{\delta},
	\label{eq:overlap_bound}
\end{equation}
\end{theorem}

\begin{IEEEproof}
Let $\mathbf{P}$ be any orthogonal projector and define $\mathbf{P}_{\perp}=\mathbf{I}-\mathbf{P}$ as the projector onto its orthogonal complement. Then each channel can be decomposed into the retained component $\mathbf{P}\mathbf{h}_k$ and the residual component $\mathbf{P}_{\perp}\mathbf{h}_k$. Since these two components lie in orthogonal subspaces, expanding $\mathbf{h}_i^H\mathbf{h}_j$ gives
\begin{equation}
	\mathbf{h}_i^H \mathbf{h}_j = (\mathbf{P}\mathbf{h}_i)^H(\mathbf{P}\mathbf{h}_j) + (\mathbf{P}_{\perp}\mathbf{h}_i)^H (\mathbf{P}_{\perp}\mathbf{h}_j).
	\label{eq:pythag}
\end{equation}
For the proposed group acquisition subspace, taking $\mathbf{P}=\mathbf{P}_{\mathcal{V}_{\mathcal{K},M}}$, the effective-channel inner product equals the inner product between the retained channel components:
\begin{equation}
	\widetilde{\mathbf{h}}_i^H\widetilde{\mathbf{h}}_j =(\mathbf{P}\mathbf{h}_i)^H(\mathbf{P}\mathbf{h}_j).
	\label{eq:effective_inner_projection}
\end{equation}
Combining~\eqref{eq:pythag} and~\eqref{eq:effective_inner_projection}, and then applying the triangle inequality and Cauchy--Schwarz inequality, yields
\begin{equation}
	|\widetilde{\mathbf{h}}_i^H \widetilde{\mathbf{h}}_j| \leq |\mathbf{h}_i^H \mathbf{h}_j| + \|\mathbf{P}_{\perp}\mathbf{h}_i\| \|\mathbf{P}_{\perp}\mathbf{h}_j\|.
	\label{eq:proof_overlap_mid}
\end{equation}
The separability condition gives $|\mathbf{h}_i^H\mathbf{h}_j| \leq \epsilon\|\mathbf{h}_i\|\|\mathbf{h}_j\|$. The capture condition gives $\|\mathbf{P}_{\perp}\mathbf{h}_k\|^2 = (1-\chi_k)\|\mathbf{h}_k\|^2 \leq \delta\|\mathbf{h}_k\|^2$. Substituting these two bounds into~\eqref{eq:proof_overlap_mid} gives
\begin{equation}
	|\widetilde{\mathbf{h}}_i^H \widetilde{\mathbf{h}}_j| \leq (\epsilon+\delta)\|\mathbf{h}_i\|\|\mathbf{h}_j\|. 
	\label{eq:proof_overlap}
\end{equation}
Moreover, $\|\widetilde{\mathbf{h}}_k\|=\|\mathbf{P}\mathbf{h}_k\| =\sqrt{\chi_k}\|\mathbf{h}_k\| \geq \sqrt{1-\delta}\|\mathbf{h}_k\|$. Thus, $\|\widetilde{\mathbf{h}}_i\|\|\widetilde{\mathbf{h}}_j\| \geq (1-\delta)\|\mathbf{h}_i\|\|\mathbf{h}_j\|$. Normalizing~\eqref{eq:proof_overlap} by $\|\widetilde{\mathbf{h}}_i\|\|\widetilde{\mathbf{h}}_j\|$ yields~\eqref{eq:overlap_bound}.
\end{IEEEproof}

\textbf{Theorem~\ref{thm:main}} gives a sufficient condition under which projection onto the group acquisition subspace preserves pairwise effective-channel separation. It does not provide a global performance guarantee, but explains the mechanism behind the proposed low-dimensional beamforming: when the selected users have low full-dimensional channel coupling and the group subspace retains most of their channel energy, the projected effective channels remain well separated. Consequently, the impact of restricting beamforming to $\mathcal{V}_{\mathcal{K},M}$ is governed by two quantities: the residual full-dimensional channel coupling after grouping and the channel energy discarded by group-subspace acquisition. This establishes the theoretical link from pre-CSI grouping and dimension-constrained CSI acquisition to stable low-dimensional RZF, which is later reflected in the conditioning CDFs and effective-rate results.
\vspace{-0.5cm}
\subsection{Online Computational Complexity}
In the conventional Type-II limited feedback pipeline, the BS first computes all pairwise normalized channel inner products over the candidate pool $\mathcal{K}_{\rm full}$, which costs $\mathcal{O}(|\mathcal{K}_{\rm full}|^2N_t)$~\cite{Yoo2006SUS}. Once these pairwise compatibility scores are available, greedy SUS selection requires $\mathcal{O}(K|\mathcal{K}_{\rm full}|)$ operations to form a size-$K$ group. The resulting full-dimensional RZF beamforming matrix then requires $\mathcal{O}(K^2N_t+K^3)$ operations for the dense matrix products and the $K\times K$ regularized inverse~\cite{Peel2005RZF,GolubVanLoan2013MatrixComputations}. At the user side, LS channel estimation from the $N_t$-port CSI-RS costs $\mathcal{O}(N_tL_c)$, while orthogonal matching pursuit (OMP)-based Type-II feedback generation costs $\mathcal{O}(O_DQN_t^2+N_tQ^2+Q^3)$ per user~\cite{3GPP38214,3GPP38215,Tropp2007OMP}. Therefore, the computational costs of BS-side and user-side processing are dominated by $\mathcal{O}(|\mathcal{K}_{\rm full}|^2N_t)$ and $\mathcal{O}((1+QO_D)N_t^2)$, respectively, when $K \ll |\mathcal{K}_{\rm full}|$, $L_c = N_t$, and $Q\ll N_t$.
\\ \indent
In the proposed framework, the BS first applies the learned predictor $\Psi_\theta$ to all RSRP fingerprints collected from $\mathcal{K}_{\rm full}$. Following the single-user site-specific Type-II design~\cite{prior_singleuser}, $\Psi_\theta$ is a three-layer MLP with hidden width $W$, so the prediction cost is $\mathcal{O}(|\mathcal{K}_{\rm full}|(BW+2W^2+2N_tQW))$. The BS then computes the predicted pairwise subspace overlaps with complexity $\mathcal{O}(|\mathcal{K}_{\rm full}|^2Q^2)$ and applies the greedy selection rule in~\eqref{eq:greedy} with $\mathcal{O}(K|\mathcal{K}_{\rm full}|)$ operations. After the group is selected, constructing the group subspace requires a partial singular value decomposition (SVD) of an $N_t\times KQ$ weighted concatenation, with cost $\mathcal{O}(N_tKQM)$~\cite{GolubVanLoan2013MatrixComputations}. The low-dimensional RZF beamforming matrix then costs $\mathcal{O}(K^2M+K^3)$, followed by the $\mathcal{O}(N_tKM)$ basis mapping to the antenna domain. At the user side, only an LS estimate of the $M$-dimensional effective channel is required, with complexity $\mathcal{O}(ML_{\mathcal{K}})$, which becomes $\mathcal{O}(M^2)$ for $L_{\mathcal{K}}=M$. Therefore, the proposed BS-side processing mainly scales as $\mathcal{O}(|\mathcal{K}_{\rm full}|(BW+W^2+N_tQW)+|\mathcal{K}_{\rm full}|^2Q^2+N_tKQM)$, while the user-side processing scales as $\mathcal{O}(M^2)$.
\\ \indent
Compared with the conventional pipeline, the reduction is twofold. At the BS, the dominant candidate-pool compatibility term is reduced from scaling with $N_t$ to scaling with the predicted subspace dimension $Q$. At the user side, the $\mathcal{O}((1+QO_D)N_t^2)$ Type-II beam search is avoided and replaced by $\mathcal{O}(M^2)$ effective-channel estimation. Hence, the main online processing is shifted from resource-limited users to the BS, where the added site-specific prediction and subspace construction are controlled by $Q,M\ll N_t$.

\section{Numerical Results}
\label{sec:results}
In this section, numerical results are provided to verify the proposed site-specific MU-MIMO framework from pre-CSI compatibility prediction to group-subspace acquisition and low-dimensional beamforming. The simulations are based on the DeepMIMO dataset~\cite{Alkhateeb2019DeepMIMO}. We mainly adopt four 3.5-GHz scenarios: \texttt{asu\_campus\_3p5} (outdoor campus, rich multipath), \texttt{boston5g\_3p5} (complex city, high diversity), \texttt{city\_0\_newyork\_3p5} (dense urban), and \texttt{city\_2\_chicago\_3p5} (regular grid, lower angular spread). Key simulation parameters are summarized in Table~\ref{tab:setup}.

\begin{table}[t]
	\centering
	\caption{Simulation Parameters}
	\label{tab:setup}
	\renewcommand{\arraystretch}{1.1}
	\small
	\begin{tabular}{lc}
		\toprule
    	Parameter & Value \\
    	\midrule
    	Carrier frequency & 3.5 GHz \\
    	Bandwidth & 10 MHz \\
    	BS antennas $N_t$ & 64 \\
    	Antenna spacing $d$ & $\lambda$/2 \\
    	SSB codebook size $B$ & 16 \\
    	User subspace dimension $Q$ & 4 \\
    	Number of candidate users $|\mathcal{K}_{\rm full}|$ & 40 \\
    	Beam-management update-rate factor $\zeta_{\rm BM}$ & 0.05 \\
    	Coherence block channel uses $T_c$ & 1000 \\
    	Transmit power $P_t$ & 40 dBm \\
    	Noise power spectrum density & $-170$ dBm/Hz \\
    	MLP depth $D$ / width $W$ & 3 / 256 \\
    	\bottomrule
	\end{tabular}
	\vspace{-0.3cm}
\end{table}

\indent The probing codebook $\mathbf{B}_\theta$ and subspace predictor $\Psi_\theta$ are trained following~\cite{prior_singleuser}. Each simulation is evaluated over 500 Monte Carlo experiments with a candidate pool of $|\mathcal{K}_{\rm full}|=40$ users drawn uniformly at random. The online pre-CSI subspace-SUS scheduler uses the deterministic initialization in~\eqref{eq:greedy_seed}. For numerical averaging, each grouping method is also repeated from 50 randomly chosen initial users in each pool to average out seed-dependent group composition. Hence, all reported metrics are averaged over the resulting $500\times 50=25{,}000$ groups per method. The beam-management update-rate factor is set to $\zeta_{\rm BM}=0.05$, which means that one SSB-based RSRP fingerprint is reused over roughly $20$ CSI-RS acquisition intervals. This reflects the timescale separation in NR: SSB-based beam management typically uses a period on the order of $20$~ms, whereas CSI-RS for CSI acquisition is configured on the order of $1$~ms~\cite{Giordani2019BeamManagementNR,3GPP38215}. 

\subsection{Evaluation of Pre-CSI Compatibility Prediction}
\label{sec:results_stage1}
We first evaluate the pre-CSI compatibility analysis in Section~\ref{sec:compatibility}, from pairwise overlap prediction to bottleneck-based group selection.

\subsubsection{Pair-Level Prediction}
At the pair level, we evaluate whether the predicted overlap $\widehat{\gamma}_{ij}$ in~\eqref{eq:gamma_hat} can correctly rank user-pair compatibility. Here, $\widehat{\gamma}_{ij}$ measures the overlap between two RSRP-predicted transmit subspaces, whereas $\gamma_{ij}^{\rm CSI}$ in~\eqref{eq:csi_pair_overlap} measures the instantaneous normalized channel overlap from true CSI. The subspace-based proxy is not required to match the instantaneous channel overlap in value. Preserving the compatibility order is sufficient for the subsequent greedy grouping. For each candidate pool, all $\binom{40}{2}=780$ user pairs are scored by these two metrics, as shown in Fig.~\ref{fig:pair_overlap_scatter}.

\begin{figure}[t]
	\centering
	\subfloat[ASU, $\rho_s=0.74$]{
		\includegraphics[width=0.44\linewidth]{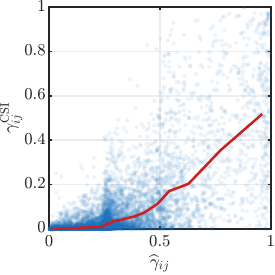}}
	\hfill
	\subfloat[Boston, $\rho_s=0.78$]{
		\includegraphics[width=0.44\linewidth]{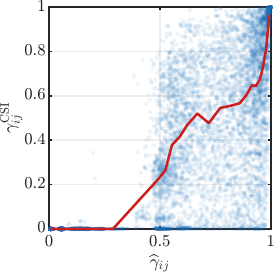}}\\[0.4em]
	\subfloat[New York, $\rho_s=0.64$]{
		\includegraphics[width=0.44\linewidth]{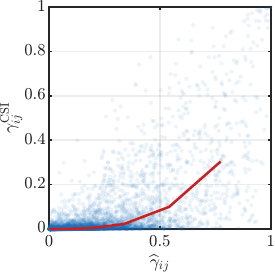}}
	\hfill
	\subfloat[Chicago, $\rho_s=0.87$]{
		\includegraphics[width=0.44\linewidth]{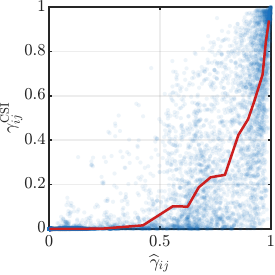}}
	\caption{Pair-level relation between the predicted subspace overlap $\widehat{\gamma}_{ij}$ and the true CSI-domain channel overlap $\gamma_{ij}^{\rm CSI}$ across four DeepMIMO scenarios. The red curves denote rolling median values over sorted predicted-overlap samples.}
	\label{fig:pair_overlap_scatter}
	\vspace{-0.5cm}
\end{figure}

\indent The Spearman rank correlation $\rho_s$ reported in each subcaption measures whether $\widehat{\gamma}_{ij}$ and $\gamma_{ij}^{\rm CSI}$ give consistent ordering of user pairs. It is positive in all scenarios, ranging from $0.64$ in New York to $0.87$ in Chicago. The rolling median curves in Fig.~\ref{fig:pair_overlap_scatter} also increase with $\widehat{\gamma}_{ij}$, indicating that pairs predicted to have larger subspace overlap tend to exhibit larger instantaneous channel overlap. The scatter is wider in dense urban scenarios, especially New York, because richer angular overlap and small-scale fading weaken the pointwise relation. These results support using $\widehat{\gamma}_{ij}$ as a pre-CSI pairwise compatibility metric, which is further tested at the group level next.

\subsubsection{Group-Level Selection}
We next test whether the bottleneck metric in~\eqref{eq:group_overlap_def} and the greedy rule in~\eqref{eq:greedy} form spatially separated MU groups, as suggested by \textbf{Proposition~\ref{prop:group_coherence}}. We compare the proposed pre-CSI subspace-SUS with random grouping, SUS with Type-II feedback, and SUS with full CSI. Each approach forms a size-$K$ group, and the resulting worst-pair overlap $\Gamma^{\rm CSI}(\mathcal{K})=\max_{i<j\in\mathcal{K}}\gamma_{ij}^{\rm CSI}$ and average overlap $\bar{\gamma}^{\rm CSI}(\mathcal{K})$ are evaluated using the true CSI. The two metrics reflect the bottleneck and average inter-user spatial coupling, respectively. Lower values indicate groups that are easier to separate by MU beamforming. $K$ is set to 4 in this experiment.

\begin{table}[t]
	\centering
	\caption{CSI-domain channel overlap evaluation after grouping with worst-pair overlap $\Gamma^{\rm CSI}$/average overlap $\bar{\gamma}^{\rm CSI}$}
	\label{tab:stage1_group}
	\renewcommand{\arraystretch}{1.1}
	\small
	\setlength{\tabcolsep}{2pt}
	\begin{tabular*}{\columnwidth}{@{}lcccc@{}}
    	\toprule
    	Scenario & Random & \textbf{Proposed} & Type-II & Full-CSI \\
    	\midrule
    	ASU & $0.317/0.090$ & $\mathbf{0.051/0.014}$ & $0.037/0.010$ & $0.007/0.002$ \\
    	Boston & $0.775/0.367$ & $\mathbf{0.174/0.032}$ & $0.023/0.005$ & $0.019/0.004$ \\
    	New York & $0.164/0.035$ & $\mathbf{0.009/0.002}$ & $0.010/0.002$ & $0.000/0.000$ \\
    	Chicago & $0.627/0.229$ & $\mathbf{0.022/0.005}$ & $0.007/0.002$ & $0.002/0.000$ \\
    	\bottomrule
	\end{tabular*}
\end{table}

Table~\ref{tab:stage1_group} shows that compared with the random selection, the proposed grouping consistently brings a substantial reduction of the true worst-pair overlap, by about $4$--$30\times$ across the four scenarios. The full-CSI SUS reference gives the best overlap as it uses perfect CSI, although a lower overlap does not necessarily imply the maximum sum-rate because SUS does not directly optimize channel strength or the final RZF rate. The Type-II baseline lies between the proposed method and this reference because it uses compressed post-CSI feedback for grouping. The proposed method remains close to Type-II grouping in most scenarios while using only pre-CSI subspace information, with the largest gap appearing in Boston, where richer scattering makes pre-CSI prediction more difficult. These results confirm that the pair- and group-level pre-CSI predictions are sufficiently reliable for forming a compatible scheduled group, which is the input to the subsequent group-specific CSI acquisition and MU beamforming stages.
\vspace{-0.3cm}
\subsection{Evaluation of Group-Subspace Capture and Conditioning}
\label{sec:results_stage2}
After group selection, the acquisition subspace in~\eqref{eq:group_subspace} should preserve the channel energy targeted by Proposition~\ref{prop:svd_opt}, while the acquired effective channels should remain well conditioned for RZF beamforming. We evaluate these two aspects separately.

\subsubsection{Channel Capture}
To validate the preservation for users' channel energy, the capture efficiency $\chi_k$ is computed according to~\eqref{eq:chi_true}. Fig.~\ref{fig:chi_vs_M} reports the average capture efficiency $\bar{\chi}$ over the scheduled users as a function of the group-subspace dimension $M$, for $K\in\{2,4,8\}$ in the four scenarios. 

\begin{figure}[t]
	\centering
	\subfloat[ASU]{
    	\includegraphics[width=0.44\linewidth]{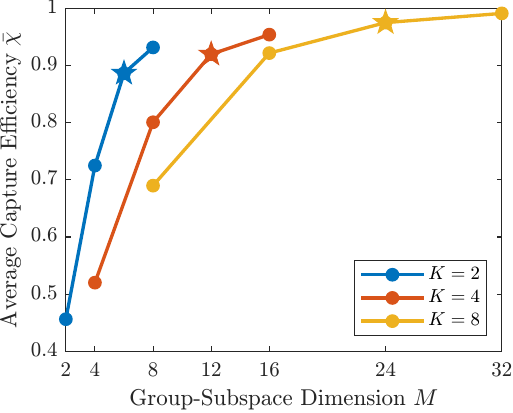}}
  	\hfill
  	\subfloat[Boston]{
    	\includegraphics[width=0.44\linewidth]{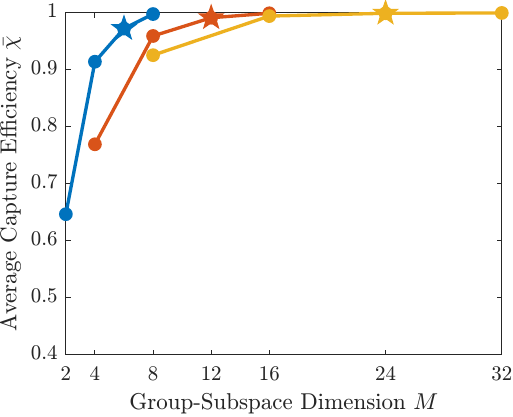}}\\
  	\subfloat[New York]{
    	\includegraphics[width=0.44\linewidth]{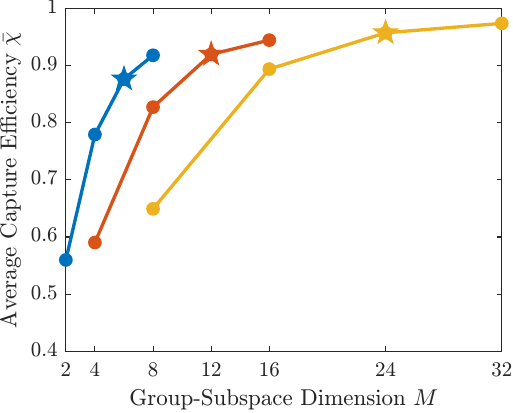}}
  	\hfill
  	\subfloat[Chicago]{
    	\includegraphics[width=0.44\linewidth]{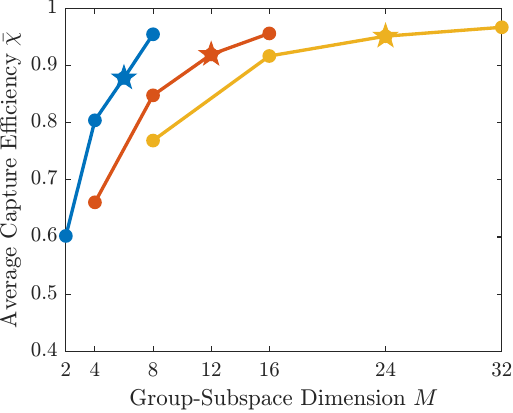}}
  	\caption{Average capture efficiency $\bar{\chi}$ versus group-subspace dimension $M$ for $K\in\{2,4,8\}$, shown per scenario. Stars highlight $M=3K$.}
  	\label{fig:chi_vs_M}
  	\vspace{-0.5cm}
\end{figure}

\indent The curves in Fig.~\ref{fig:chi_vs_M} show a consistent trend across all scenarios: $\bar{\chi}$ increases rapidly as $M$ grows and then gradually saturates. At $M=3K$, the average capture efficiency is close to or above $0.9$ in all cases. The gain from further increasing $M$ to $4K$ is small, indicating that the SVD-based group subspace has already retained most of the dominant channel energy. Thus, $M=3K$ serves as a robust rate-preserving operating point for group-subspace acquisition, while a smaller $M$ may still be preferable when online overhead dominates. The scenario dependence is mainly visible at smaller $M$: Boston is easier to capture, whereas ASU, New York, and Chicago require a larger group subspace depending on $K$.

\subsubsection{Effective-Channel Conditioning}
High capture efficiency alone does not guarantee reliable MU beamforming, because the projected user channels must also remain sufficiently separated as characterized by \textbf{Theorem~\ref{thm:main}}. We therefore evaluate $\log_{10}\kappa(\widetilde{\mathbf{G}})$, where $\widetilde{\mathbf{G}}$ is the normalized effective Gram matrix defined in~\eqref{eq:gram}. The experiment uses $K=4$ and $M=3K$ and compares random grouping, the proposed pre-CSI grouping, full-CSI SUS, and the ideal-subspace reference. ``Ideal Subspace'' denotes a reference that keeps the proposed pre-CSI selected group but constructs the acquisition and beamforming subspace from the true dominant user subspaces rather than the RSRP-predicted ones. It is used only to separate subspace-prediction loss from the subsequent low-dimensional processing.

\begin{figure}[t]
	\centering
	\subfloat[ASU]{
		\includegraphics[width=0.44\linewidth]{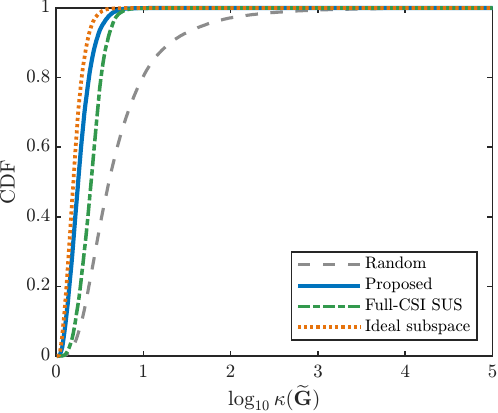}}
	\hfill
	\subfloat[Boston]{
		\includegraphics[width=0.44\linewidth]{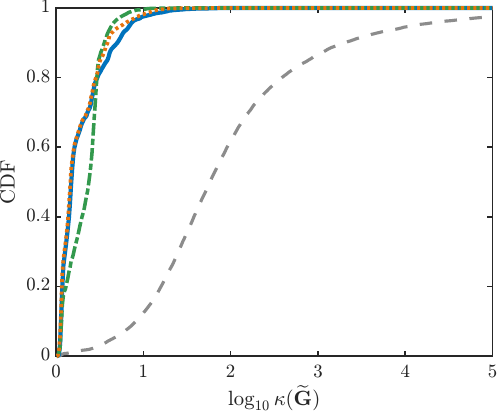}}\\
	\subfloat[New York]{
		\includegraphics[width=0.44\linewidth]{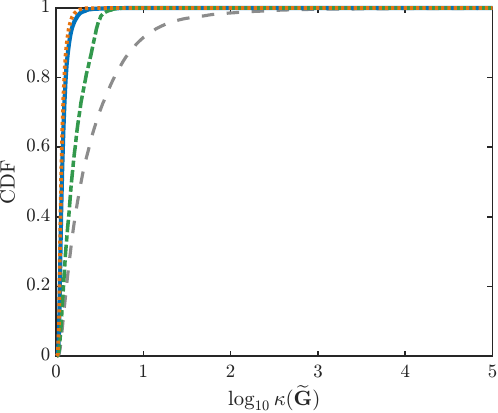}}
	\hfill
	\subfloat[Chicago]{
		\includegraphics[width=0.44\linewidth]{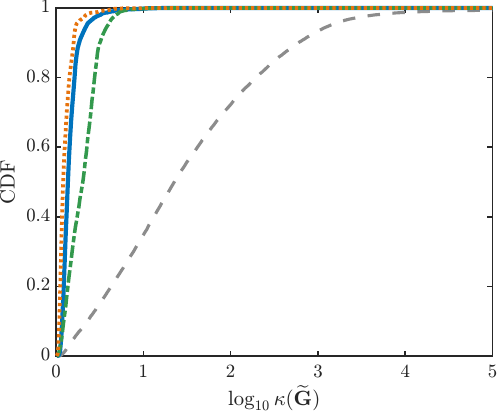}}
	\caption{CDF of the effective-channel conditioning metric 	$\log_{10}\kappa(\widetilde{\mathbf{G}})$.}
	\label{fig:cond_cdf}
	\vspace{-0.5cm}
\end{figure}

Fig.~\ref{fig:cond_cdf} shows that random grouping has a pronounced right-tail, especially in Boston and Chicago, indicating that some scheduled groups lead to nearly dependent effective channels after projection. In contrast, the proposed method yields cumulative distribution function (CDF) curves concentrated near the left side of the axis and closely follows the full-CSI SUS and ideal-subspace references in all scenarios. This confirms that the proposed grouping and group-subspace construction not only retain most of the channel energy, but also preserve a well-conditioned effective channel geometry for low-dimensional RZF beamforming.

\subsection{Evaluation of Effective Rate}
\label{sec:results_stage4}
Finally, we evaluate whether compatibility inference, group acquisition, and low-dimensional RZF translate into the overhead-aware objective in~\eqref{eq:ref_prob}. $K=8$ is set for this evaluation. The achievable MU sum-rate $R_{\rm sum}$ is first reported to isolate the rate loss without overhead. We then evaluate the final effective rate $R_{\rm eff}$ as a function of the coherence block length $T_c$, which directly reflects the tradeoff between CSI acquisition overhead and data transmission.
\begin{figure}[t]
	\centering
	\subfloat[ASU]{
		\includegraphics[width=0.44\linewidth]{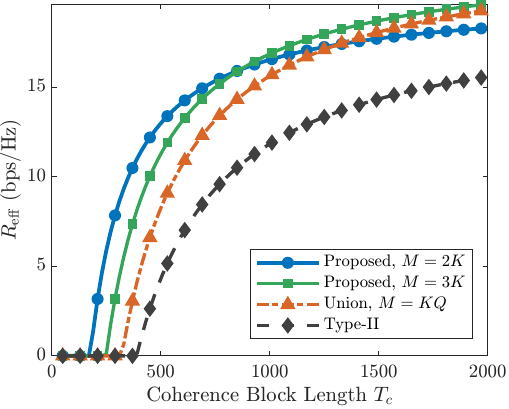}}
	\hfill
	\subfloat[Boston]{
		\includegraphics[width=0.44\linewidth]{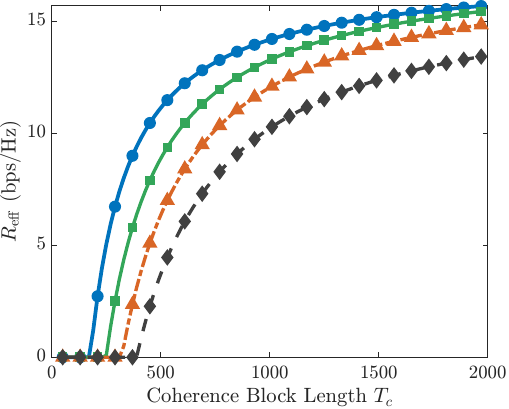}}\\
	\subfloat[New York]{
		\includegraphics[width=0.44\linewidth]{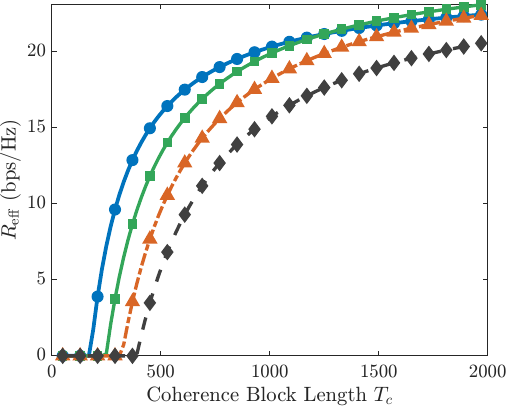}}
	\hfill
	\subfloat[Chicago]{
		\includegraphics[width=0.44\linewidth]{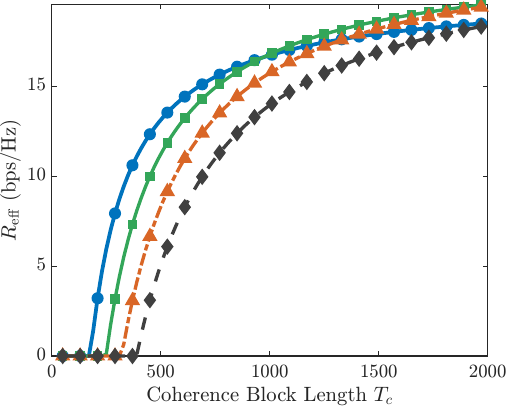}}
	\caption{Effective-rate comparison versus coherence block length $T_c$.}
	\label{fig:reff_vs_tc}
	\vspace{-0.3cm}
\end{figure}
\begin{table}[t]
	\centering
	\caption{Comparison of achievable MU sum-rate}
	\label{tab:rsum_peel_k8}
	\renewcommand{\arraystretch}{1.12}
	\setlength{\tabcolsep}{3pt}
	\begin{tabular*}{\columnwidth}{@{\extracolsep{\fill}}lcccc}
		\toprule
		Scheme & ASU & Boston & NewYork & Chicago \\
		\midrule
		Full-CSI SUS             & 22.23 & 20.44 & 25.48 & 24.38 \\
		Post-Group CSI  & 23.24 & 17.74 & 27.06 & 24.24 \\
		Union Subspace    & 23.03 & 17.73 & 26.73 & 23.18 \\
		\textbf{\boldmath Proposed, $M=3K$}          & \textbf{22.46} & \textbf{17.67} & \textbf{26.41} & \textbf{22.35} \\
		\textbf{\boldmath Proposed, $M=2K$}          & \textbf{20.11} & \textbf{17.23} & \textbf{24.64} & \textbf{20.32} \\
		Type-II                   & 19.38 & 16.73 & 25.60 & 22.84 \\
		\bottomrule
	\end{tabular*}
	\vspace{-0.5cm}
\end{table}

Table~\ref{tab:rsum_peel_k8} separates the raw-rate impact of grouping, subspace prediction, and dimension compression. ``Full-CSI SUS'' uses perfect CSI for grouping and full-dimensional RZF, but it is not a strict upper bound because SUS favors channel orthogonality rather than directly maximizing the final RZF sum-rate; hence ``Post-Group CSI'' can slightly exceed it in ASU and New York. ``Post-Group CSI'' isolates the grouping effect by using the proposed pre-CSI group with perfect post-group CSI, while ``Union Subspace'' further includes the loss from replacing full antenna-domain CSI by the predicted user-subspace union. The proposed method with $M=3K$ stays within $1$~bps/Hz of ``Union Subspace'' in all scenarios, showing that the compact group subspace is largely rate-preserving. The $M=2K$ setting incurs a larger raw-rate loss but reduces CSI-RS and feedback dimensions, making it more suitable for overhead-limited regimes.
\\ \indent Fig.~\ref{fig:reff_vs_tc} evaluates the rate-overhead tradeoff versus coherence block length. The proposed framework consistently outperforms Type-II feedback, confirming the benefit of reducing both CSI-RS and feedback dimensions. For short and moderate $T_c$, $M=2K$ often gives the highest $R_{\rm eff}$ because its lower overhead compensates for the raw-rate loss in Table~\ref{tab:rsum_peel_k8}. As $T_c$ increases, the rate-preserving setting $M=3K$ catches up with or exceeds $M=2K$ in several scenarios. The union-subspace reference has competitive raw rate but is less effective due to its $KQ$-dimensional acquisition overhead. Thus, $M=2K$ is preferable in overhead-limited regimes, while $M=3K$ provides a more rate-preserving operating point.

\section{Conclusion}
\label{sec:conclusion}
This paper developed a site-specific MU-MIMO beamforming framework that exploits SSB-stage RSRP fingerprints to reduce the CSI acquisition and feedback burden of limited-feedback MU transmission. The BS first infers user subspaces from RSRP fingerprints, selects a compatible group through pre-CSI subspace-overlap metrics, and then constructs a compact group CSI-RS subspace for effective-channel feedback and low-dimensional RZF beamforming. Theoretical analysis links predicted subspace overlap to CSI-domain channel coupling, establishes the projection-energy optimality of the group-subspace construction, and relates channel capture to effective-channel conditioning. Numerical results on DeepMIMO scenarios verify this design chain and show that the proposed low-overhead acquisition achieves higher effective rates than Type-II feedback in short and moderate coherence blocks while substantially reducing user-side processing.
\\ \indent 
Future work may extend the proposed framework to rank-adaptive multi-antenna users. Another important direction is online adaptation of the scheduled group size and acquisition dimension according to coherence length, traffic load, and prediction uncertainty. Extending the site-specific acquisition principle to multi-cell interference-aware coordination is also a promising direction.

\balance
\bibliographystyle{IEEEtran}
\bibliography{refs}

@ARTICLE{Heath2016mmWave,
	author={Heath, Robert W. and Gonz{\'a}lez-Prelcic, Nuria and Rangan, Sundeep and Roh, Wonil and Sayeed, Akbar M.},
	journal={IEEE J. Sel. Topics Signal Process.},
	title={An Overview of Signal Processing Techniques for Millimeter Wave {MIMO} Systems},
	year={2016},
	month={Feb.},
	volume={10},
	number={3},
	pages={436-453},
	doi={10.1109/JSTSP.2016.2523924}}

@ARTICLE{Christensen2008WMMSE,
	author={Christensen, Søren Skovgaard and Agarwal, Rajiv and De Carvalho, Elisabeth and Cioffi, John M.},
	journal={IEEE Trans. Wireless Commun.},
	title={Weighted Sum-Rate Maximization Using Weighted {MMSE} for {MIMO-BC} Beamforming Design},
	year={2008},
	month={Dec.},
	volume={7},
	number={12},
	pages={4792-4799},
	doi={10.1109/T-WC.2008.070851}}

@ARTICLE{Sun2018NOMAMIMO,
	author={Sun, Xiaofang and Yang, Nan and Yan, Shihao and Ding, Zhiguo and Ng, Derrick Wing Kwan and Shen, Chao and Zhong, Zhangdui},
	journal={IEEE Trans. Wireless Commun.},
	title={Joint Beamforming and Power Allocation in Downlink {NOMA} Multiuser {MIMO} Networks},
	year={2018},
	month={Aug.},
	volume={17},
	number={8},
	pages={5367-5381}}

@ARTICLE{Zhao2026PRAFD,
	author={Zhao, Chengjie and Gong, Yuanzhe and Le-Ngoc, Tho},
	journal={IEEE Trans. Wireless Commun.},
	 title={Joint Positioning, Beamforming, and Power Allocation in Full-Duplex MIMO With Position-Reconfigurable Antenna Arrays}, 
	year={2026},
	volume={25},
	month={Jun.},
	pages={18750-18763},
	doi={10.1109/TWC.2026.3703296}}

@ARTICLE{Love2008LimitedFeedback,
	author={Love, David J. and Heath, Robert W. and N. Lau, Vincent K. and Gesbert, David and Rao, Bhaskar D. and Andrews, Matthew},
	journal={IEEE J. Sel. Areas Commun.},
	title={An Overview of Limited Feedback in Wireless Communication Systems},
	year={2008},
	month={Oct.},
	volume={26},
	number={8},
	pages={1341-1365},
	doi={10.1109/JSAC.2008.081002}}

@ARTICLE{Larsson2014MassiveMIMO,
	author={Larsson, Erik G. and Edfors, Ove and Tufvesson, Fredrik and Marzetta, Thomas L.},
	journal={IEEE Commun. Mag.},
	title={Massive {MIMO} for Next Generation Wireless Systems},
	year={2014},
	month={Feb.},
	volume={52},
	number={2},
	pages={186-195},
	doi={10.1109/MCOM.2014.6736761}}

@ARTICLE{Spencer2004ZF,
	author={Spencer, Quentin H. and Swindlehurst, A. Lee and Haardt, Martin},
	journal={IEEE Trans. Signal Process.},
	title={Zero-Forcing Methods for Downlink Spatial Multiplexing in Multiuser {MIMO} Channels},
	year={2004},
	month={Feb.},
	volume={52},
	number={2},
	pages={461-471},
	doi={10.1109/TSP.2003.821107}}

@ARTICLE{Yoo2006SUS,
	author={Yoo, Taesang and Goldsmith, Andrea},
	journal={IEEE J. Sel. Areas Commun.},
	title={On the Optimality of Multiantenna Broadcast Scheduling Using Zero-Forcing Beamforming},
	year={2006},
	month={Feb.},
	volume={24},
	number={3},
	pages={528-541},
	doi={10.1109/JSAC.2006.872709}}

@ARTICLE{Gonzalez1985FarthestFirst,
	author={Gonzalez, Teofilo F.},
	journal={Theor. Comput. Sci.},
	title={Clustering to Minimize the Maximum Intercluster Distance},
	year={1985},
	volume={38},
	pages={293-306},
	doi={10.1016/0304-3975(85)90224-5}}

@ARTICLE{Jindal2006Feedback,
	author={Jindal, Nihar},
	journal={IEEE Trans. Inf. Theory},
	title={{MIMO} Broadcast Channels With Finite-Rate Feedback},
	year={2006},
	month={Nov.},
	volume={52},
	number={11},
	pages={5045-5060},
	doi={10.1109/TIT.2006.882421}}

@ARTICLE{Wagner2012RZF,
	author={Wagner, Sebastian and Couillet, Romain and Debbah, M{\'e}rouane and Slock, Dirk T. M.},
	journal={IEEE Trans. Inf. Theory},
	title={Large System Analysis of Linear Precoding in Correlated {MISO} Broadcast Channels Under Limited Feedback},
	year={2012},
	month={Jul.},
	volume={58},
	number={7},
	pages={4509-4537},
	doi={10.1109/TIT.2012.2192300}}

@techreport{3GPPTR38901,
	author      = {{3GPP}},
	title       = {Study on Channel Model for Frequencies from 0.5 to 100 {GHz}},
	institution = {3rd Generation Partnership Project (3GPP)},
	type        = {Technical Report},
	number      = {TR 38.901},
	year        = {2022}
}

@techreport{3GPP38211,
	author      = {{3GPP}},
	title       = {{NR}; Physical Channels and Modulation},
	institution = {3rd Generation Partnership Project (3GPP)},
	type        = {Technical Specification},
	number      = {TS 38.211},
	year        = {2018}
}

@techreport{3GPP38214,
	author      = {{3GPP}},
	title       = {{NR}; Physical Layer Procedures for Data},
	institution = {3rd Generation Partnership Project (3GPP)},
	type        = {Technical Specification},
	number      = {TS 38.214},
	year        = {2018}
}

@techreport{3GPP38215,
	author      = {{3GPP}},
	title       = {{NR}; Physical Layer Measurements},
	institution = {3rd Generation Partnership Project (3GPP)},
	type        = {Technical Specification},
	number      = {TS 38.215},
	year        = {2018}
}

@ARTICLE{Giordani2019BeamManagementNR,
	author={Giordani, Marco and Polese, Michele and Roy, Arnab and Castor, Douglas and Zorzi, Michele},
	journal={IEEE Commun. Surv. Tutorials},
	title={A Tutorial on Beam Management for {3GPP NR} at mm{W}ave Frequencies},
	year={2019},
	month={Sep},
	volume={21},
	number={1},
	pages={173-196},
	doi={10.1109/COMST.2018.2869411}}

@ARTICLE{Heng2022SiteSpecificProbing,
	author={Heng, Yuqiang and Mo, Jianhua and Andrews, Jeffrey G.},
	journal={IEEE Trans. Wireless Commun.},
	title={Learning Site-Specific Probing Beams for Fast mm{W}ave Beam Alignment},
	year={2022},
	month={Jan.},
	volume={21},
	number={8},
	pages={5785-5800},
	doi={10.1109/TWC.2022.3143121}}

@ARTICLE{Zeng2020CKM,
	author={Zeng, Yong and Xu, Xiaoli},
	journal={IEEE Wireless Commun.}, 
	title={Toward Environment-Aware {6G} Communications via Channel Knowledge Map}, 
	year={2021},
	month={Mar},
	volume={28},
	number={3},
	pages={84-91},
	doi={10.1109/MWC.001.2000327}}

@ARTICLE{Alkhateeb2018DLBeam,
	author={Alkhateeb, Ahmed and Alex, Sam and Varkey, Paul and Li, Ying and Qu, Qi and Tujkovic, Djordje},
	journal={IEEE Access},
	title={Deep Learning Coordinated Beamforming for Highly-Mobile Millimeter Wave Systems},
	year={2018},
	volume={6},
	pages={37328-37348},
	doi={10.1109/ACCESS.2018.2850226}}

@INPROCEEDINGS{Ning2023RSRPCodebook,
	author={Ning, Xinzhi and Zhang, Shutao and Xue, Ye and Zheng, Xi and Shi, Qingjiang and Chang, Tsung-Hui},
	booktitle={Proc. {IEEE} Int. Workshop Signal Process. Adv. Wireless Commun. (SPAWC)},
	title={Learning Beams Adaptive to the Environment: An {RSRP}-Based Codebook Design},
	year={2023},
	pages={521-525},
	doi={10.1109/SPAWC53906.2023.10304486}}

@article{sim,
	title={Generative Site-Specific Beamforming via Information-Maximizing Codebook},
	author={Zhao, Cheng-Jie and Wang, Zhaolin and Liu, Yuanwei},
	journal={arXiv preprint arXiv:2602.12552},
	year={2026}
}

@article{prior_singleuser,
	author={Zhao, Cheng-Jie and Wang, Zhaolin and Zhao, Zongyao and Liu, Yuanwei},
	title={Bridging Standardized Codebook and Site-Specific Beamforming: A Unified Limited-Feedback Framework},
	journal={arXiv preprint arXiv:2604.14524},
	year={2026}
}

@ARTICLE{Adhikary2013JSDM,
	author={Adhikary, Ansuman and Nam, Junyoung and Ahn, Jae-Young and Caire, Giuseppe},
	journal={IEEE Trans. Inf. Theory}, 
	title={Joint Spatial Division and Multiplexing—The Large-Scale Array Regime}, 
	year={2013},
	month={Oct.},
	volume={59},
	number={10},
	pages={6441-6463},
	doi={10.1109/TIT.2013.2269476}}

@ARTICLE{Peel2005RZF,
	author={Peel, Christian B. and Hochwald, Bertrand M. and Swindlehurst, A. Lee},
	journal={IEEE Trans. Commun.},
	title={A Vector-Perturbation Technique for Near-Capacity Multiantenna Multiuser Communication---{P}art {I}: Channel Inversion and Regularization},
	year={2005},
	month={Jan.},
	volume={53},
	number={1},
	pages={195-202},
	doi={10.1109/TCOMM.2004.840354}}

@ARTICLE{Ayach2014SpatiallySparse,
	author={Ayach, Omar El and Rajagopal, Sridhar and Abu-Surra, Shadi and Pi, Zhouyue and Heath, Robert W.},
	journal={IEEE Trans. Wireless Commun.},
	title={Spatially Sparse Precoding in Millimeter Wave {MIMO} Systems},
	year={2014},
	month={Mar.},
	volume={13},
	number={3},
	pages={1499-1513},
	doi={10.1109/TWC.2014.011714.130846}}

@INPROCEEDINGS{Alkhateeb2019DeepMIMO,
	author={Alkhateeb, Ahmed},
	booktitle={Proc. Inf. Theory Appl. Workshop (ITA)},
	title={{DeepMIMO}: A Generic Deep Learning Dataset for Millimeter Wave and Massive {MIMO} Applications},
	year={2019},
	month={Feb.},
	pages={1-8},
	address={San Diego, CA}}

@article{SiFo,
	title={{SiFo}: Wireless Foundation Model for Low-Overhead Site-Specific {CSI} Feedback},
	author={Zhao, Cheng-Jie and Wang, Zhaolin and Zhao, Zongyao and Liu, Yuanwei},
	journal={arXiv preprint arXiv:2605.16141},
	year={2026}
}

@ARTICLE{Zhao2025LDS,
	author={Zhao, Xiaotong and Shi, Qingjiang},
	journal={IEEE Trans. Signal Process.},
	title={A Universal Low-Dimensional Subspace Structure in Beamforming Design: Theory and Applications},
	year={2025},
	month={Apr.},
	volume={73},
	pages={1775-1791}}

@book{HornJohnson2012MatrixAnalysis,
	author={Horn, Roger A. and Johnson, Charles R.},
	title={Matrix Analysis},
	edition={2},
	publisher={Cambridge University Press},
	address={Cambridge, U.K.},
	year={2012}}

@book{GolubVanLoan2013MatrixComputations,
	author={Golub, Gene H. and Van Loan, Charles F.},
	title={Matrix Computations},
	edition={4},
	publisher={Johns Hopkins University Press},
	address={Baltimore, MD, USA},
	year={2013}}

@article{Tropp2007OMP,
	author={Tropp, Joel A. and Gilbert, Anna C.},
	journal={IEEE Trans. Inf. Theory},
	title={Signal Recovery From Random Measurements Via Orthogonal Matching Pursuit},
	year={2007},
	month={Dec.},
	volume={53},
	number={12},
	pages={4655-4666},
	doi={10.1109/TIT.2007.909108}}

\end{document}